\documentclass[letterpaper, 12pt]{article}

\usepackage{amssymb,amsmath,amsthm}
\usepackage{epsf}
\usepackage{times,bm,mathrsfs}
\usepackage{amsfonts}
\usepackage{mathrsfs}
\usepackage{bbm}
\usepackage{color}
\usepackage{graphicx}
\usepackage{amscd}
\usepackage{epsfig}
\usepackage{enumitem}
\usepackage{comment}
\usepackage{sectsty}

\definecolor{Red}{rgb}{1,0,0}
\definecolor{Blue}{rgb}{0,0,1}
\definecolor{Olive}{rgb}{0.41,0.55,0.13}
\definecolor{Green}{rgb}{0,1,0}
\definecolor{MGreen}{rgb}{0,0.8,0}
\definecolor{DGreen}{rgb}{0,0.55,0}
\definecolor{Yellow}{rgb}{1,1,0}
\definecolor{Cyan}{rgb}{0,1,1}
\definecolor{Magenta}{rgb}{1,0,1}
\definecolor{Orange}{rgb}{1,.5,0}
\definecolor{Violet}{rgb}{.5,0,.5}
\definecolor{Purple}{rgb}{.75,0,.25}
\definecolor{Brown}{rgb}{.75,.5,.25}
\definecolor{Grey}{rgb}{.5,.5,.5}
\definecolor{Black}{rgb}{0,0,0}

\usepackage{thmtools}
\usepackage{thm-restate}

\usepackage{hyperref}

\usepackage{cleveref}

\usepackage{indentfirst}

\newtheorem{defi}{Definition}

\newtheorem{algo}{Algorithm}
\newtheorem{prop}{Proposition}

\newtheorem{lemma}{Lemma}

\newtheorem{claim}{Claim}
\newtheorem{example}{Example}
\theoremstyle{remark}

\newcommand{\T}{\mathcal{T}}

\newcommand{\Net}{\mathcal{N}}
\newcommand{\Network}{\mathcal{N} = (\mathcal{X}, \mathcal{S}, w)}
\newcommand{\Taxa}{\mathcal{X}}
\newcommand{\Splits}{\mathcal{S}}

\newcommand{\ellipseC}{\Delta + \Omega + \tau}
\newcommand{\ellipseR}{3\Delta + 7\Omega + 8\tau}
\newcommand{\M}{3\Delta + 7\Omega + 10\tau}
\newcommand{\connectD}{\Delta + 2\Omega + \tau}

\newcommand{\AllSmallSplits}{\hat{\mathcal{S}}|_{B(x,y)}}
\newcommand{\SmallSplit}{\hat{S}|_{C(x,y)}}
\newcommand{\SmallSplits}{\hat{\mathcal{S}}|_{C(x,y)}}
\newcommand{\BigSplit}{\hat{S}|_{x,y}}
\newcommand{\BigSplits}{\hat{\mathcal{S}}|_{x,y}}

\newcommand{\AllSmallSplitsO}{\mathcal{S}|_{B(x,y)}}

\newcommand{\SmallSplitsO}{\mathcal{S}|_{C(x,y)}}

\newcommand{\BigSplitsO}{\mathcal{S}|_{x,y}}

\newcommand{\Split}[1]{S^{(#1)} = \{S^{(#1)}_1, S^{(#1)}_2\}}

\newcommand{\Comp}{\mathscr{C}}
\newcommand{\Inco}{\mathscr{I}}

\newcommand{\sep}{\Sigma}
\newcommand{\nsep}{\Gamma}

\newcommand{\connects}[3]{{#1}\overset{#3}{\leftrightharpoons}{#2}}

\author{Sebastien Roch\footnote{Department of Mathematics at the University of Wisconsin--Madison.
		Work supported by NSF grants DMS-1007144,  DMS-1149312 (CAREER), and DMS-1614242.} 
		\and 
		Kun-Chieh Wang\footnote{Supported by NSF grant DMS-1149312 and DMS-1149312 to SR.} }
\date{\today}
\title{Circular Networks from Distorted Metrics}

\addtocounter{page}{-1}

\begin{document}

\maketitle
\begin{abstract}
Trees have long been used as a graphical representation of species relationships. However complex evolutionary events, such as genetic reassortments or hybrid speciations which occur commonly in viruses, bacteria and plants, do not fit into this elementary framework. Alternatively, various network representations have been developed. Circular networks are a natural generalization of leaf-labeled trees interpreted as split systems, that is, collections of bipartitions over leaf labels corresponding to current species. Although such networks do not explicitly model specific evolutionary events of interest, their straightforward visualization and fast reconstruction have made them a popular exploratory tool to detect network-like evolution in genetic datasets.

Standard reconstruction methods for circular networks, such as Neighbor-Net, rely on an associated metric on the species set. Such a metric is first estimated from DNA sequences, which leads to a key difficulty: distantly related sequences produce statistically unreliable estimates. This is problematic for Neighbor-Net as it is based on the popular tree reconstruction method Neighbor-Joining, whose sensitivity to distance estimation errors is well established theoretically. In the tree case, more robust reconstruction methods have been developed using the notion of a distorted metric, which captures the dependence of the error in the distance through a radius of accuracy. Here we design the first circular network reconstruction method based on distorted metrics. Our method is computationally efficient. Moreover, the analysis of its radius of accuracy highlights the important role played by the maximum incompatibility, a measure of the extent to which the network differs from a tree.
\end{abstract}

\thispagestyle{empty}

  \clearpage

\section{Introduction}

Trees have long been used to represent species relationships~\cite{Felsenstein:04,Steel:16,Warnow:u}. The leaves of a phylogenetic
tree correspond to current species while its branchings indicate past speciation events. However, complex evolutionary events, such as genetic reassortments or hybrid speciations, do not fit into this elementary framework. Such non-tree-like events play an important role
in the evolution of viruses, bacteria and plants. This issue has led to the development of various notions of {\em phylogenetic networks}~\cite{Book}. 

A natural generalization of phylogenetic trees is obtained by representing them as split networks, that is, collections of bipartitions over the species set. 
On a tree whose leaves are labeled by species names,
each edge can be thought of as a bipartition over the
species: removing the edge produces 
exactly two connected components. 
In this representation,
trees are characterized by the fact that their splits have a certain compatibility property~\cite{SempleSteel:03}. More generally, circular networks
relax this compatibility property, while retaining enough structure to be useful as representations of evolutionary history~\cite{SplitDecomposition}. Such networks are widely used in practice.
Although they do not explicitly model specific evolutionary events, their straightforward visualization and fast reconstruction have made them a popular exploratory tool to detect network-like evolution in genetic datasets~\cite{Huson01022006}. They are also
useful in cases where data is insufficient to 
single out a unique tree-like history, but instead supports 
many possible evolutionary scenarios.

Standard reconstruction methods for circular networks, such as the Neighbor-Net algorithm introduced in~\cite{NeighborNet}, rely on a metric on the species set. Such a metric, which quantifies how far apart species are in the Tree of Life, is estimated from genetic data. Very roughly, it counts
how many mutations separate any two species. This leads to a key difficulty: under standard stochastic models
of DNA evolution, distantly related sequences are known to produce statistically unreliable distance estimates~\cite{ErStSzWa:99a,ErStSzWa:99b}. This is problematic for Neighbor-Net, in particular, as it is based on the popular tree reconstruction method Neighbor-Joining, whose sensitivity to distance estimation errors is well established theoretically~\cite{LaceyChang:06}. 

In the tree case, more robust reconstruction methods were developed using the notion of a distorted metric which captures the dependence of the error in the distance through a radius of accuracy~\cite{03KZZ,Mossel:07}. 
A key insight to come out of this line of work,
starting with the seminal results of~\cite{ErStSzWa:99a,ErStSzWa:99b},
is that a phylogenetic tree 
can be reconstructed using only a subset of
the pairwise distances---those less than roughly the
chord depth of the tree. Here the chord depth of an edge is
the shortest path between two leaves passing through
that edge and the chord depth of the tree is the maximum
depth among its edges. This result is remarkable because,
in general,
the depth can be significantly smaller than
the diameter. As a consequence,
a number of results have been obtained showing that,
under common stochastic models of sequence evolution,
a polynomial amount of data suffices to reconstruct
a phylogenetic tree with bounded branch lengths. See e.g.~\cite{CrGoGo:02,MosselRoch:06,TreeCase,RSA:RSA20372}.
This approach has also inspired practical reconstruction methods~\cite{DCM,DCM3}.

Here we design the first reconstruction method for circular networks based on distorted metrics. In addition to generalizing the chord depth, we show that, unlike the tree case, pairwise distances within the chord depth do not
in general suffice to reconstruct these networks. We introduce the
notion of maximum incompatibility, a measure of the extent to which the network differs from a tree,
to obtain a tight (up to a constant) bound on the
required radius of accuracy. Before stating our main results,
we provide some background on split networks.

%---------------------------------NN and SD----------------------------------
\subsection{Background}

We start with some basic definitions. 
See~\cite{Book} for an in-depth exposition.
\begin{defi}[Split networks~\cite{SplitDecomposition}]
A {\bf split} $S = (S_1, S_2)$ on a set of taxa $\Taxa$ is an unordered bipartition of $\Taxa$ into two non-empty, disjoint sets: $S_1, S_2 \in \Taxa$, $S_1\cap S_2 =\emptyset$, $S_1\cup S_2 = \Taxa$.
We say that $\Net = (\Taxa,\Splits, w)$ is a {\bf weighted split network} (or split network for short) on a set of $\Taxa$ if $\Splits$ is a set of splits on $\Taxa$ and $w: \Splits \to (0,\infty)$ is a positive split weight function. We assume that any two splits $\Split{1}$, $\Split{2}$ in $\Splits$ are distinct, that is, $S^{(1)}_1 \neq S^{(2)}_1, S^{(2)}_2$.
\end{defi}
\noindent For any $x, y \in \Taxa$, we let
$\Splits|_{x,y}$ be the collection of splits in $\Splits$
separating $x$ and $y$, that is,
$$
\Splits|_{x,y}
= \{S \in \Splits\,:\, \delta_S(x,y) = 1\},
$$
where $\delta_S(x,y)$, known as the split metric, is the indicator of whether $S = (S_1, S_2)$ separates $x$ and $y$
\begin{eqnarray}
	\delta_S(x,y) = \left\{
	\begin{array}{ll}
		0, & \text{if } x,y\in S_1\text{ or }x,y\in S_2. \\
		1. & \text{otherwise}.
	\end{array}
	\right.
	\label{eq:separation-metric}
\end{eqnarray}
For a split $S \in \Splits|_{x,y}$, we write
$S = \{S_x, S_y\}$ where $x \in S_x$ and $y \in S_y$.
For simplicity, we assume that
$\Splits|_{x,y} \neq \emptyset$ for all $x,y \in \Taxa$. (Taxa not separated by a split can be identified.)

Let $T = (V,E)$ be a binary tree with leaf set $\Taxa$ and non-negative edge weight function $w : E \to [0,+\infty)$. We refer to  
$\T = ( \Taxa, V, E, w)$ as a phylogenetic tree. 
Any phylogenetic tree can be represented as a weighted split network. For each edge $e \in E$, 
define a split on $\Taxa$ as follows: after deleting $e$,
the vertices of $\T$ form two disjoint connected components with corresponding leaf sets $S^1$ and $S^2$;  we let $S_e = \{S^1, S^2\}$ be the split  generated by $e$ in this way. Conversely, one may ask: given a split network $\Net = (\Taxa, \Splits, w)$, is there a phylogenetic tree $\T = ( \Taxa, V, E, w)$ such that $\Splits = \{S_e: e\in E\}$? To answer this question, we need the concept of compatibility.
\begin{defi}[Compatibility \cite{Compatibility}]
Two splits $\Split{1}$ and $\Split{2}$ are called {\bf compatible}, if at least one of the following intersections is empty:
\begin{eqnarray*}
S_1^{(1)}\cap S_1^{(2)}, \quad S_1^{(1)}\cap S_2^{(2)}, \quad S_2^{(1)}\cap S_1^{(2)}, \quad S_2^{(1)}\cap S_2^{(2)}.
\end{eqnarray*} 
We write $S^{(1)} \sim S^{(2)}$ to indicate that
$S^{(1)}$ and $S^{(2)}$ are compatible.
Otherwise, we say that the two splits are {\bf incompatible}. A set of splits $\Splits$ is called compatible if all pairs of splits in $\Splits$ are compatible.
\end{defi}
\noindent In words, for any two splits, there is one side of one and one side of the other that are disjoint. The following result was first proved in~\cite{Compatibility}.
Given a split network $\Net = (\Taxa, \Splits, w)$, there is a phylogenetic tree $\mathcal{T} = ( \Taxa, V, E, w)$ such that $\Splits = \{S_e: e\in E\}$ if and only if $\Splits$ is compatible. For a collection of 
splits $S^{(1)}, \ldots, S^{(\ell)}$ on $\Taxa$, we let
\begin{equation}
\label{eq:def-compatible}
\Comp_\Net(S^{(1)}, \ldots, S^{(\ell)})
= \{S \in \Splits\,:\, S \sim S^{(i)}, \forall i\},
\end{equation}
be the set of splits of $\Net$ compatible with all
splits in $S^{(1)}, \ldots, S^{(\ell)}$, and we let
\begin{equation}
\label{eq:def-incompatible}
\Inco_\Net(S^{(1)}, \ldots, S^{(\ell)})
= \{S \in \Splits\,:\, \exists i, S \nsim S^{(i)}\},
\end{equation}
be the set of splits of $\Net$ incompatible with at least one
split in $S^{(1)}, \ldots, S^{(\ell)}$. We drop the subscript $\Net$ when the network is clear from context.

Most split networks cannot be realized as phylogenetic trees. The following is an important special class of more general split networks.
\begin{defi}[Circular networks~\cite{SplitDecomposition}]
A collection of splits $\Splits$ on $\Taxa$ is called {\bf circular} if there exists a linear ordering $(x_1, \dots, x_n)$ of the elements of $\Taxa$ for $\Splits$ such that each split $S\in \Splits$ has the form:
\begin{eqnarray*}
S = \{\: \{x_{p}, \dots, x_{q}\}\:, \:\Taxa - \{x_{p}, \dots, x_{q}\}\:\}
\end{eqnarray*}
for $1 < p \leq q \leq n$. We say that a split network $\Net = \{\Taxa, \Splits, w\}$ is a {\bf circular network} if $\Splits$ is circular.
\end{defi}
\noindent Phylogenetic trees, seen as split networks, are special cases of circular networks (e.g.~\cite{Book}). Circular networks have
the appealing feature that they cannot contain
too many splits. Indeed, let $\Network$ be 
a circular network with $|\Taxa| = n$. Then
$|\Splits| = O(n^2)$~\cite{SplitDecomposition}.
In general, circular networks are harder to interpret than
trees are. In fact, they are not meant to represent
explicit evolutionary events. However,
they admit an appealing vizualization in the
form of an outer-labeled (i.e., the taxa are
on the outside) planar graph that gives some
insight into how ``close to a tree'' the network is.
As such, they are popular exploratory analysis tools.
We will not describe this vizualization and how it is used here,
as it is quite involved. See, e.g.,~\cite[Chapter 5]{Book}
for a formal definition and~\cite{Huson01022006} for examples of applications.

Split networks are naturally associated with
a metric. We refer to a function $d:\Taxa\times \Taxa\to [0,+\infty]$ as a {\bf dissimilarity} over $\Taxa$ if it is symmetric and
$d(x,x) = 0$ for all $x$.
\begin{defi}[Metric associated to a split network]
	\label{metric}
	Let $\Net = (\Taxa, \Splits,w)$  be a split network.
	The dissimilarity $d:\Taxa\times \Taxa\to [0,\infty)$ defined as follows
	\begin{eqnarray*}
		d(x,y)=\sum_{S\in\Splits|_{x,y}} w(S),
	\end{eqnarray*}
	for all
	$x,y \in \Taxa$, is referred to as the metric
	associated to $\Net$.
	(It can be shown that
	$d$ is indeed a metric. In particular, it satisfies the triangle inequality.)
\end{defi} 
\noindent The metric associated with a circular network
can be used to reconstruct it. 
\begin{defi}[$d$-splits]
Let $d:\Taxa\times \Taxa\to [0,\infty)$ be a dissimilarity.
The {\bf isolation index} $\alpha_d(S)$ of a split $S = \{S_1, S_2\}$ over $\Taxa$
is given by
\begin{eqnarray*}
	\alpha_d(S)=\min\{\tilde{\alpha}_d(x_1,y_1|x_2,y_2)\,:\, x_1,y_1\in S_1, x_2,y_2\in S_2\},
\end{eqnarray*}
where
\begin{eqnarray*}
	\tilde{\alpha}_d(x_1,y_1|x_2,y_2) &=& \frac{1}{2}(\max\{d(x_1,y_1)+d(x_2,y_2), d(x_1,x_2)+d(y_1,y_2), \\
	&& \qquad\qquad  d(x_1,y_2)+d(y_1,x_2)\}-d(x_1,y_1)-d(x_2,y_2)).
\end{eqnarray*}
(Note that the latter is always non-negative.)
We say that $S$ is a {\bf $d$-split} if $\alpha_d(S) > 0$.
\end{defi}
\noindent The following result establishes that
circular networks can be reconstructed from
their associated metric.
\begin{lemma}[$d$-splits and circular networks~\cite{SplitDecomposition}]
\label{lemma:dsplits-weakcomp}
Let $\Taxa$ be a set of $n$ taxa and
let $\Network$ be a circular network with
associated metric $d$. Then $\Splits$ coincides with the
set of all $d$-splits of $\Network$. Further the isolation index $\alpha_d(S)$ equals $w(S)$ for all $S\in\Splits$. 
\end{lemma}
\noindent The {\bf split decomposition method}, described in Section~\ref{Algorithm}, reconstructs $\Network$ from $d$ in polynomial time. When $\Net$ is compatible, $d$ is an {\bf additive metric}. See e.g.~\cite{SempleSteel:03,Steel:16}.

In practice one obtains an estimate $\hat{d}$ of $d$, called the {\bf distance matrix}, from DNA sequences, e.g., through the Jukes-Cantor formula~\cite{JukesCantor}
or the log-det distance~\cite{Steel:94}. The accuracy of this estimate depends on the amount of data used~\cite{ErStSzWa:99a,ErStSzWa:99b}. In previous work in the context of tree reconstruction,
distorted metrics were used to  encode the fact that large $d$-values typically produce unreliable $\hat{d}$-estimates.
\begin{defi}[Distorted metrics~\cite{03KZZ,Mossel:07}]
	Suppose $\Network$ is a split network with associated metric $d$. Let $\tau, R>0$. We say that a dissimilarity $\hat{d}: \Taxa\times \Taxa\to[0,+\infty]$ is a {\bf $(\tau,R)$-distorted metric} of $\Net$ if $\hat{d}$ is accurate on ``short'' distances, that is, for all $x,y\in \Taxa$
		\begin{eqnarray*}
			d(x,y)<R+\tau \quad \text{or}\quad  \hat{d}(x,y)<R+\tau \quad \implies \quad |d(x,y)-\hat{d}(x,y)|<\tau.
		\end{eqnarray*} 
	We refer to $\tau$ and $R$ as the {\bf tolerance}
	and {\bf accuracy radius} of $\hat{d}$ respectively.
\end{defi}
\noindent Distorted metrics have previoulsy been motivated by analyzing Markov models on trees that are commonly used to model the evolution of DNA sequences~\cite{ErStSzWa:99a,ErStSzWa:99b}. Such
models have also been extended to split networks~\cite{GroupBasedModel}. 

\subsection{Main results}

By the reconstruction result mentioned above,
any circular network $\Network$ with associated metric $d$ can be reconstructed from a $(\tau,R)$-distorted metric
where $\tau$ is $0$ and $R$ is greater or equal than the diameter $\max\{d(x,y)\,:\, x,y \in \Taxa\}$ of $\Net$.
In the tree case, it has been shown that {\em a much smaller} $R$ suffice~\cite{ErStSzWa:99a,ErStSzWa:99b,Mossel:07,TreeCase}.
Here we establish such results for circular networks.

\paragraph{Chord depth and maximum incompatibility}
To bound the tolerance and accuracy radius
needed to reconstruct a circular
network from a distorted metric, we introduce 
several structural parameters.
The first two parameters generalize naturally from the tree context.
\begin{defi}[Minimum weight]
Let $\Network$ be a split network. The {\bf minimum
	weight} of $\Net$ is given by
$$
\epsilon_\Net 
= \min\{w(S)\,:\,S \in \Splits\}.
$$
\end{defi}
\noindent Let $\Network$ be a split network with associated metric $d$. For a subset of splits $\mathcal{A} \subseteq \Splits$, we let
\begin{equation}
\label{eq:restricted}
d(x,y;\mathcal{A}) = \sum_{S \in \Splits|_{x,y} \cap \mathcal{A}} w(S),
\end{equation}
be the distance between $x$ and $y$ restricted to those
splits in $\mathcal{A}$.
\begin{defi}[Chord depth]
	Let $\Network$ be a split network with associated metric $d$.
	The {\bf chord depth} of a split $S \in \Splits$ is
	\begin{eqnarray*}
		\Delta_\Net(S) = \min\left\{d(x,y;\Comp_\Net(S)) \,:\, x,y\in\Taxa, S \in \Splits|_{x,y} \right\},
	\end{eqnarray*}
	and the {\bf chord depth} of $\Net$ is the largest chord depth among all of its  splits
	\begin{eqnarray*}
		\Delta_\Net = \max\left\{\Delta_\Net(S): S\in\Splits \right\}.
	\end{eqnarray*}
\end{defi}
\noindent It was shown in~\cite[Corollary 1]{TreeCase} that,
if $\Network$ is compatible, then a $(\tau, R)$-distorted metric 
with $\tau < \frac{1}{4} \epsilon_\Net$ and $R > 2 \Delta_\Net + \frac{5}{4}\epsilon_\Net$ suffice to reconstruct $\Net$ in
polynomial time (among compatible networks).

For more general circular networks, 
the minimum weight and
chord depth are not sufficient to characterize
the tolerance and accuracy radius required
for reconstructibility; see Example~\ref{ex:max-inco} below. For that purpose, we introduce a new notion that, roughly speaking, measures
the extent to which a split network differs from a tree.
\begin{defi}[Maximum incompatibility]
	Let $\Network$ be a split network.
	The {\bf incompatible weight} of a split $S\in\Splits$ is
	\begin{eqnarray*}
		\Omega_\Net(S) = \sum_{S' \in \Inco(S)} w(S'),
	\end{eqnarray*}
	and the {\bf maximum incompatibility} of $\Net$ is the largest incompatible weight among all of its splits
	\begin{eqnarray*}
		\Omega_\Net = \max\{\Omega_\Net(S): S\in\Splits\}.
	\end{eqnarray*}
\end{defi}
\noindent We drop the subscript in $\epsilon_\Net$,
$\Delta_\Net$ and $\Omega_\Net$ when the $\Net$
is clear from context. 

\paragraph{Statement of results}
We now state our main result.
\begin{restatable}{theorem}{NetworkReconstruction}
	\label{NetworkReconstructionLabel}
	Suppose $\Network$ is a circular network. Given a $(\tau,R)$-distorted metric with $\tau < \frac{1}{4}\epsilon_\Net$ and $R > 3\Delta_\Net + 7 \Omega_\Net + \frac{5}{2} \epsilon_\Net$, the split set $\Splits$
	can be reconstructed in polynomial time together 
	with weight estimates $\hat{w} \,:\, \Splits \to (0,+\infty)$ satisfying $|\hat{w}(S)-w(S)| < 2\tau$.
\end{restatable}
\noindent Establishing robustness to noise
of circular network reconstruction algorithms
is important given that, as explained above, such networks
are used in practice to tentatively diagnose deviations from tree-like evolution. Errors due to noise can confound such analyses. See e.g.~\cite{Huson01022006} for a discussion of these issues.

In~\cite[Section 4]{TreeCase}, it was shown that 
in the tree case the accuracy radius
must depend linearly on the depth.
The following example shows that \emph{the accuracy radius
must also depend linearly on the maximum
incompatibility.}
\begin{figure}[!b]
	\centering
	\includegraphics[scale=0.2]{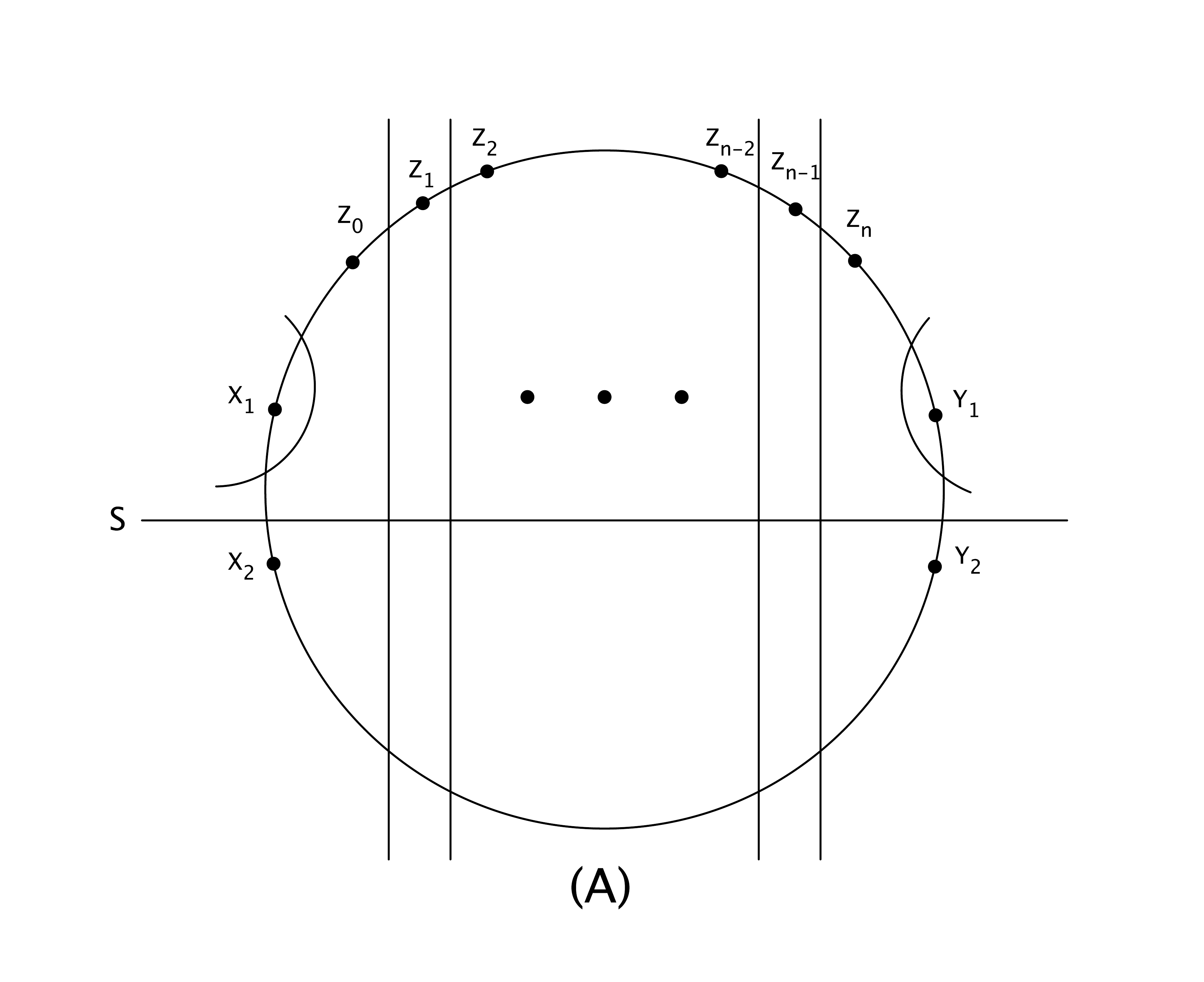}
	\includegraphics[scale=0.2]{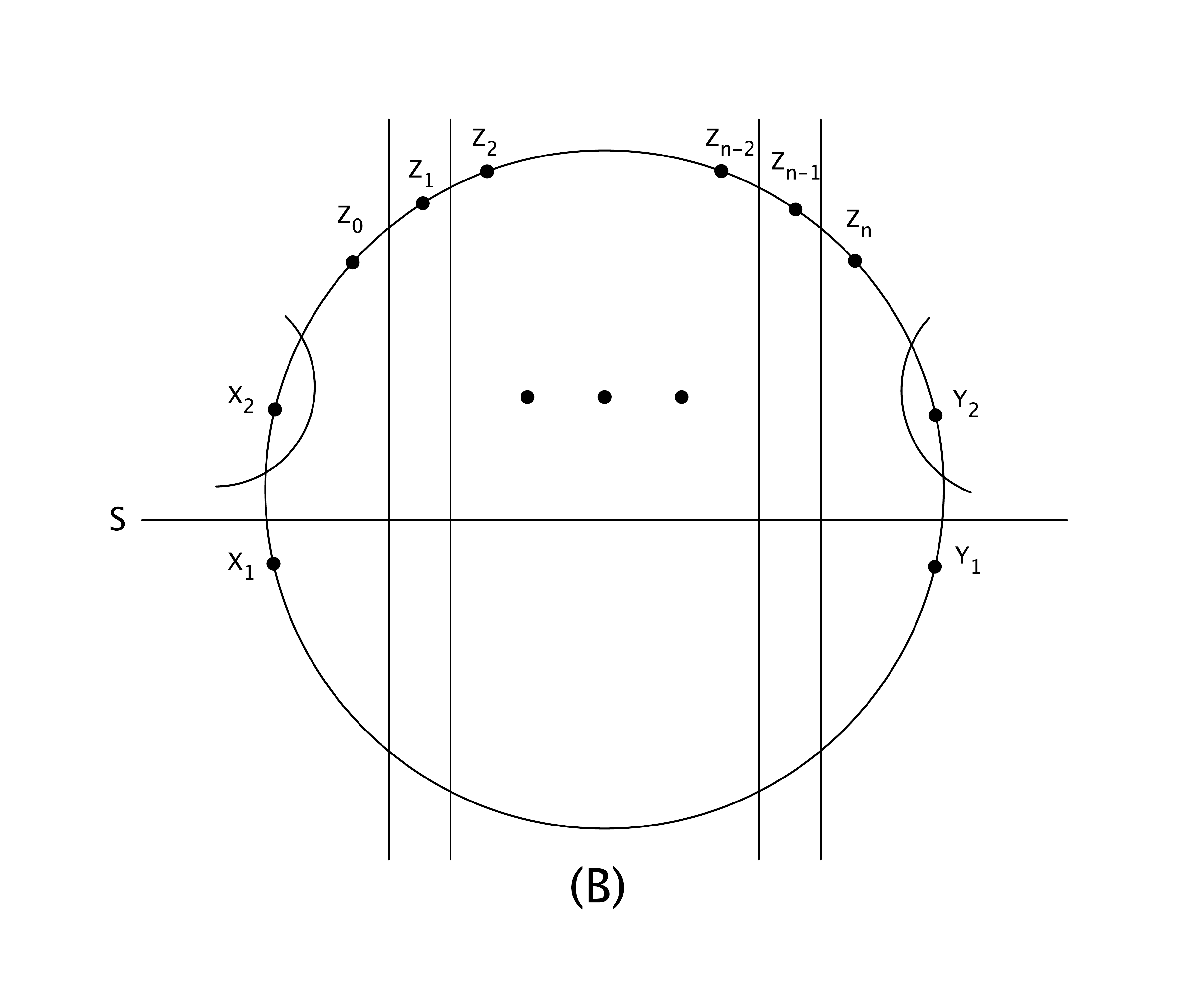}
	\caption{Two circular networks indistinghishable 
		from a distorted metric with sublinear
		dependence on the maximum incompatibility.
		Here the taxa are ordered on a circle and lines indicate splits. For instance, in (A), the leftmost vertical line is the split with $\{z_0, x_1, x_2\}$ on one side and all other taxa on the other.
		In both networks, $\Taxa = \{x_1, x_2, y_1, y_2\} \cup \{ \cup_i z_i\}$, and the $n$ vertical lines, the horizontal line, and the two arcs are splits  of weight 1. }
	\label{LowerBoundSecondClaim}
\end{figure}
\begin{example}[Depth is insufficient; linear dependence in maximum
	incompatibility is needed]
\label{ex:max-inco}
Consider the two circular networks in \Cref{LowerBoundSecondClaim}. In both networks, $\Taxa = \{x_1, x_2, y_1, y_2\} \cup \{ \cup_i z_i\}$, and the $n$ vertical lines, the horizontal line, and the two arcs are splits of weight 1. The chord depth of both networks is $1$ while their maximum incompatibility is $n$. In both networks 
\begin{itemize}
	\item[-] $d(z_i, x_j) = i +1$, $0\leq i \leq n$, $1\leq j \leq 2$,
	\item[-] $d(z_i, y_j) = n-i+1$, $0\leq i \leq n$, $1\leq j \leq 2$,
	\item[-] $d(x_1, x_2) = d(y_1,y_2) = 2$,
	\item[-] $d(x_1,y_2) = d(x_2, y_1) = n+2$.
\end{itemize}
The only difference is that, in graph (A), $d(x_1, y_1) = n+2$ and $d(x_2, y_2) = n$ while, in graph (B), $d(x_2, y_2) = n+2$ and $d(x_1, y_1) = n$. If we choose the distance matrix $\hat{d}$ as follows:
\begin{itemize}
	\item[-] $\hat{d}(x_1, y_1) = \hat{d}(x_2,y_2) = n+1$,
	\item[-] $\hat{d} = d$ for all other pairs,
\end{itemize}
then $\hat{d}$ is a $(\tau, n-1)$-distorted metric of both networks for any $\tau\in (0,1)$. Hence, these two circular networks are indistinguishable from $\hat{d}$. Observe that the chord depth is 1 for any $n$, but the maximum incompatibility can be made arbitrary large.
(Note that the claim still holds if we replace
the chord depth with the ``full chord depth''
$\max\{\min\{d(x,y)\,:\,x,y \in \Taxa, S \in \Splits|_{x,y}\}\,:\, S \in \Splits\}$, which also includes weights of incompatible splits separating $x$ and $y$.)
\end{example}

\paragraph{Proof idea}
Our proof of Theorem~\ref{NetworkReconstructionLabel}
is based on a divide-and-conquer approach of~\cite{TreeCase},
first introduced in~\cite{Mossel:07} and also related to the seminal work of~\cite{ErStSzWa:99a,ErStSzWa:99b} on short quartet methods and the decomposition methods of~\cite{DCM,DCM3}.
More specifically, we first reconstruct sub-networks in regions
of small diameter. We then extend the bipartitions to the full
taxon set by hopping back from each taxon to this small region and recording
which side of the split is reached first. However, the
work of~\cite{TreeCase} relies heavily on the tree
structure, which simplifies many arguments. Our novel contributions here are twofold:
\begin{itemize}
	\item We define the notion of maximum incompatibility and highlight its key role in the reconstruction of circular networks, as we discussed above.
	
	\item We extend the effective divide-and-conquer methodology developed in~\cite{ErStSzWa:99a,ErStSzWa:99b,DCM,DCM3,Mossel:07,TreeCase} to circular networks. The analysis of this more general class of split networks is more involved than the tree case. In particular, we introduce the notion of a compatible chain---an analogue of paths in graphs---which may be of independent interest in the study of split networks. %See Section~\ref{sec:key-lemmas}.
	
\end{itemize}

%\begin{remark}
%In fact our results apply to the more general class
%of weakly compatible networks. See, e.g.,~\cite{Book}
%for background on such networks. Details are left out.
%\end{remark}

\subsection{Organization}

The rest of the paper is organized as follows. The algorithm is detailed in~\Cref{Algorithm}. The proof of our main theorem can be found in \Cref{Analysis}. %Most details are left to the appendix. 

%---------------------------------split decomposition Method ----------------------------------

%---------------------------------Algorithm ----------------------------------

\section{Algorithm}
\label{Algorithm}

In this section, we describe our reconstruction
algorithm. 

\paragraph{Split decomposition}
One building block of our reconstruction algorithm is the split
decomposition method of~\cite{SplitDecomposition},
which is detailed below.
\begin{algo}
	\label{algo:SD}
	Split decomposition method~\cite{SplitDecomposition}\\
	Given a distance matrix $d$ on $\Taxa=\{z_1, z_2,\dots,z_n\}$, compute the set of all $d$-splits on $\Taxa$ as follows:\\
	Initially, set $\Taxa_1 = \{ z_1\}$ and $\Splits_1 = \emptyset$. Assume we have the set of all $d$-splits $\Splits_i$ on the first $i$ taxa $\Taxa_i=\{z_1, z_2,\dots,z_i\}$. To obtain $\Splits_{i+1}$ on $\Taxa_{i+1}=\{z_1, z_2,\dots,z_{i+1}\}$, for each split $S = \{ S_1, S_2 \}\in\Splits_i$:
	\begin{itemize}
		\item If $\alpha_d(\{S_1\cup\{z_{i+1}\},S_2\})>0$, then add $\{S_1\cup\{z_{i+1}\},S_2\}$ to $\mathcal{S}_{i+1}$. 
		\item If $\alpha_d(\{S_1,S_2\cup\{z_{i+1}\}\})>0$, then add $\{S_1,S_2\cup\{z_{i+1}\}\}$ to $\mathcal{S}_{i+1}$. 
		\item If $\alpha_d(\{\Taxa_i,\{z_{i+1}\}\})>0$, then add $\{\Taxa_i,\{z_{i+1}\}\}$ to $\mathcal{S}_{i+1}$.  
	\end{itemize}
	The result is given by $\mathcal{S}_n$.
\end{algo}
\begin{lemma}[Split decomposition method~\cite{SplitDecomposition}]
	\label{lemma:split-decomp}
Let $d$ be a dissimilarity on $\Taxa$ with $|\Taxa| = n$.
The split decomposition method applied to $d$ 
is guaranteed to return exaclty the set of
$d$-splits in polynomial time.
\end{lemma}

\paragraph{From small regions to full networks}
As in~\cite{TreeCase}, our reconstruction algorithm is composed of two main parts: Mini Reconstruction and Bipartition Extension. The purpose of Mini Reconstruction is to uncover the splits in ``small regions,'' for which we use the split decomposition method. In Bipartition Extension, each split found in a small region is extended to a split on all taxa by recursively adding taxa to the split. The main steps of the algorithm
are illustrated in~\Cref{NotationOverview}.
The input to the algorithm is a distorted metric of
a circular network $\Net$
together with bounds on $\Delta_\Net$, $\Omega_\Net$ and $\epsilon_\Net$. The output is a collection of weighted splits, that is, a split network.
\begin{figure}[!t]
  \centering
  	\includegraphics[scale=0.7]{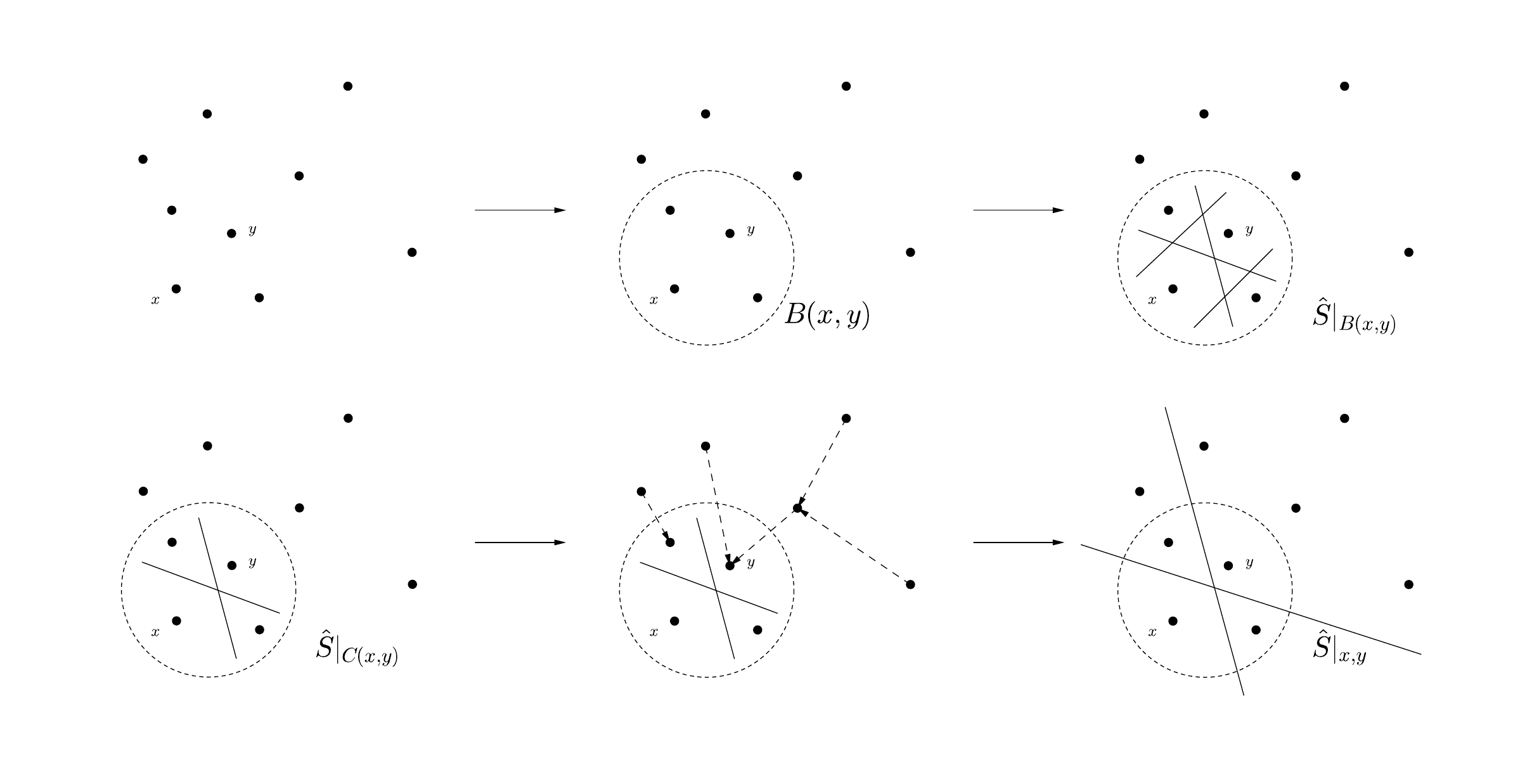}
	\caption{Illustration of the main steps of the reconstruction algorithm. On the top row, Mini Reconstruction uncovers all splits in a small region $B(x,y)$ around $x,y$. On the bottom row, those splits separating $x$ and $y$ in $B(x,y)$ are extended to full splits using Bipartition Extension.}
	\label{NotationOverview}
\end{figure}

\begin{algo}
\label{MainAlgorithm}
Network reconstruction\\
{\bf Main loop:}
\hrule
\noindent Input: $\Taxa$, $\tau$, $\Delta$, $\Omega$ and a $(\tau,R)$-distorted metric $\hat{d}$ with $R > 3\Delta + 7\Omega + 10\tau $\\
Output: A set of splits $\hat{\Splits}$ on $\Taxa $ and a weight function $\hat{w}: \hat{\Splits} \to (0, \infty)$
\hrule
\begin{enumerate}
\item Initially $\hat{\Splits} := \emptyset$.
\item Let EllipseRadius $:= \ellipseR$, ConnectingDistance $ := \connectD$.
\item For all pair taxa $x,y \in \Taxa$ satisfying $\hat{d}(x,y) \leq \ellipseC$:\\
\indent \hspace{7.5mm} $\SmallSplits$, $\hat{w}|_{C(x,y)}$ := MiniReconstruction($\Taxa$ ,$x$, $y$, $\tau$, EllipseRadius,  $\hat{d}$) \\
\indent \hspace{7.5mm} $\BigSplits$, $\hat{w}|_{x,y}$ := BipartitionExtension($\Taxa$, $\SmallSplits$, $\hat{w}|_{C(x,y)}$, ConnectingDistance, $\hat{d}$)\\
\indent\hspace{7.5mm} Set $\hat{\Splits} := \hat{\Splits}\cup\BigSplits$ and $w(\BigSplit) := \hat{w}|_{x,y}(\BigSplit)$ for any $\BigSplit \in \BigSplits$
\item Return $\hat{\Splits}$, $\hat{w}$.
\end{enumerate}

\noindent {\bf MiniReconstruction:}
\hrule
\noindent Input: $\Taxa$, $x$, $y$, $\tau$, EllipseRadius, $\hat{d}$\\
Output: A set of splits $\SmallSplits$ on $B(x,y)$ and a weight function $\hat{w}|_{C(x,y)}: \SmallSplits \to (0, \infty)$ 
\hrule
\begin{enumerate}
\item Set $B(x,y) := \{ z \in \Taxa: \hat{d}(z,x)+\hat{d}(z,y) \leq EllipseRadius \} $.
\item Apply the split decomposition method, \Cref{algo:SD}, to find all $\hat{d}$-splits on $B(x,y)$ with isolation indices greater than $2 \tau$. Denote that collection of splits by $\AllSmallSplits$ and let $\hat{w}|_{B(x,y)}: \AllSmallSplits \to (0, \infty)$ be the corresponding isolation indices.
\item Set $\SmallSplits := \{\SmallSplit \in \AllSmallSplits: \delta_{\SmallSplit}(x,y) = 1\}$ and for any $\SmallSplit\in\SmallSplits$, set $\hat{w}|_{C(x,y)}(\SmallSplit) := \hat{w}|_{B(x,y)}(\SmallSplit)$ .
\item Return $\SmallSplits$, $\hat{w}|_{C(x,y)}$.
\end{enumerate}

\noindent{\bf BipartitionExtension:}
\hrule
\noindent Input: $\Taxa$, $\SmallSplits$, $\hat{w}|_{C(x,y)}$, ConnectingDistance, $\hat{d}$ \\
Output: A set of splits $\BigSplits$ on $\Taxa$ and a weight function $\hat{w}|_{x,y}: \BigSplits \to (0, \infty)$ 
\hrule
\begin{enumerate}
\item Initially, set $\BigSplits = \emptyset$.
\item For all $S =\{S_1,S_2\} \in \SmallSplits$:\\
\indent \hspace{7.5mm} Set $w_0 := \hat{w}|_{C(x,y)}(S)$\\
\indent \hspace{7.5mm} While $\Taxa - (S_1\cup S_2) \neq \emptyset$:\\
\indent \hspace{15mm} Find $i \in{1, 2}$, $z \in S_i$, $z' \in \Taxa - (S_1\cup S_2)$ such that $\hat{d}(z,z') \leq $ ConnectingDistance\\
\indent \hspace{15mm} Set $S_i := S_i \cup \{z'\}$\\
\indent \hspace{7.5mm} Set $\BigSplits := \BigSplits \cup \{S\}$, $\hat{w}|_{x,y}(S) := w_0$ 
\item Return $\BigSplits$, $\hat{w}|_{x,y}$.
\end{enumerate}

\end{algo}

%---------------------------------Analysis----------------------------------

\section{Analysis}
\label{Analysis}

In this section, we prove Theorem~\ref{NetworkReconstructionLabel}.
In the remainder of this section, $\Network$ is a circular network with minimum weight $\epsilon$,
chord depth $\Delta$ and maximum incompatibility $\Omega$.
We assume that $\hat{d}$ is a $(\tau,R)$-distorted metric of $\Net$ with $\tau < \frac{1}{4}\epsilon$ and
$R > 3 \Delta + 7 \Omega + \frac{5}{2}\epsilon  > \M$.
For any $x, y \in \Taxa$ with $\hat{d}(x,y) \leq \ellipseC$, we let $B(x, y)$ be the ``small region'' 
\begin{eqnarray*}
	B(x,y) = \{ z \in \Taxa: \hat{d}(z,x)+\hat{d}(z,y) \leq \ellipseR \}.
\end{eqnarray*}
We denote by $\AllSmallSplits$ the set of all $\hat{d}$-split over $B(x,y)$ which are found via the split decomposition method to have isolation index larger than $2\tau$ in the Mini Reconstruction sub-routine of the algorithm. Let $\SmallSplits$ be the subset of $\AllSmallSplits$ containing those splits separating $x$ and $y$
\begin{eqnarray*}
\SmallSplits = \{ S \in \AllSmallSplits : \delta_{\SmallSplit}(x,y) = 1\}.
\end{eqnarray*}
The Bipartition Extension sub-routine extends each split in $\SmallSplits$ to a split $\BigSplit$ over all taxa, the collection of which we denote by $\BigSplits$. The algorithm outputs
\begin{eqnarray*}
\hat{\Splits} = \cup_{x,y\in\Taxa : \hat{d}(x,y) \leq \ellipseC} \:\BigSplits.
\end{eqnarray*}
See~\Cref{NotationOverview} for an illustration. 

To analyze the correctness of the reconstruction
algorithm, we also let
\begin{equation*}
	\AllSmallSplitsO = \{ \{S_1\cap B(x,y), S_2 \cap B(x,y)\} : S = \{S_1, S_2\} \in \Splits, S_1\cap B(x,y) \neq \emptyset, S_2 \cap B(x,y) \neq \emptyset\},
\end{equation*}
and
\begin{equation*}
	\SmallSplitsO = \{ \{S_x\cap B(x,y), S_y \cap B(x,y)\} : S = \{S_x, S_y\} \in \BigSplitsO, S_x\cap B(x,y) \neq \emptyset, S_y \cap B(x,y) \neq \emptyset\},
\end{equation*}
where recall that 
$\BigSplitsO$ is the set of splits separating $x$ and $y$ in $\Net$.
We will establish the following claims:
\begin{enumerate}[label=(\Alph*)]
\item Mini Reconstruction correctly reconstructs the splits restricted to $B(x,y)$.
\begin{prop}[Correctness of Mini Reconstruction]
	\label{prop:mini-correct}
For all $x,y\in\Taxa$ satisfying $\hat{d}(x,y) \leq \ellipseC$, we have $\AllSmallSplits = \AllSmallSplitsO$.
\end{prop}

\item Bipartition Extension correctly extends 
the splits separating $x$ and $y$ in $B(x,y)$
to all of $\Taxa$.
\begin{prop}[Correctness of Bipartition Extension]
	\label{prop:ext-correct}
For all $x,y\in\Taxa$ satisfying $\hat{d}(x,y) \leq \ellipseC$, we have $\BigSplits = \BigSplitsO$.
\end{prop}

\item All splits are reconstructed. 
\begin{prop}[Exhaustivity]
	\label{prop:exhaust}
$\Splits = \cup_{x,y\in\Taxa : \hat{d}(x,y) \leq \ellipseC}\:\BigSplitsO$, so  $\hat{\Splits} = \Splits$.
\end{prop}

\item Isolation indices computed by the split decomposition method are good estimates of split weights.
\begin{prop}[Weight estimates]
	\label{prop:weight-estimates}
	For any $S \in \BigSplits$, if $S'$ is the corresponding split on $B(x,y)$, then we have $|w(S) - \alpha_{\hat{d}}(S')| < 2\tau$.
\end{prop}
\end{enumerate}

%---------------------------------Distance Lemmas----------------------------------

\subsection{Key distance lemmas}

We begin with a few structural results
that will play a key role in the proof. The proofs are %in the appendix.
in Section~\ref{sec:appendix}.

\paragraph{Witnesses.} 
Recall from~\eqref{eq:restricted} the definition
of the restricted distance and from~\eqref{eq:def-compatible} and~\eqref{eq:def-incompatible} the definitions
of the sets of splits compatible or incompatible with a given split.
For $S \in \Splits$,
we refer to a pair of taxa $x,y$ such that $S \in \Splits|_{x,y}$ and
$d(x,y;\Comp(S)) \leq \Delta$ as {\bf $S$-witnesses}. 
The next lemma establishes the existence of witnesses and gives a bound on the distance between
them. 
\begin{lemma}[Witnesses]
	\label{WitnessOfDepth}
	Let $\Network$ be a split
	network with chord depth $\Delta$ and
	maximum incompatibility $\Omega$.
	For all split $S\in\Splits$, there exists 
	a pair $x,y\in \Taxa$ of $S$-witnesses. Moreover, 
	for all such pairs, $d(x,y)\leq\Delta+\Omega$.
\end{lemma}

\paragraph{Hoppability.}
For $\kappa > 0$, we say that $x, y\in \Taxa$ are {\bf $\kappa$-hoppable} if there exists a sequence of taxa $\{z_i\}_{i=0}^{\ell+1} \subseteq \Taxa$ such that $z_0 = x$, $z_{\ell+1} = y$ and, for any $0\leq i\leq \ell$, we have $d(z_i, z_{i+1}) \leq \kappa$.
In that case, we write $\connects{x}{y}{\kappa}$
and we refer to the pairs $(z_i,z_{i+1})$ as {\bf $\kappa$-hops}.
We say that $\Net$ is {\bf $\kappa$-hoppable} if $\connects{x}{y}{\kappa}$ for $\forall x, y\in \Taxa$. 
Our goal is to establish $\kappa$-hoppability 
for the smallest possible $\kappa$. This is
the most involved step of the proof.
\begin{lemma}[Hoppability]
	\label{lemma:hoppability}
	\label{connectable}
	Let $\Network$ be a split
	network with chord depth $\Delta$ and
	maximum incompatibility $\Omega$. Then
	$\Net$ is $\kappa$-hoppable
	with $\kappa \geq \Delta + 2\Omega$.
\end{lemma}
\noindent %Example~\ref{ex:omega-tight} in the Appendix
The following example
shows that the constants in Lemma~\ref{lemma:hoppability} are tight. 
\begin{example}[Tightness of $\Omega$ factor]
	\label{ex:omega-tight}
	By definition of the chord depth, $\kappa$ in Lemma~\ref{lemma:hoppability} must be at least $\Delta$. Indeed, otherwise, the split achieving $\Delta$ cannnot be crossed.
	The following example shows that the constant factor in the $\Omega$-term is also tight. Consider the graph in~\Cref{ConnectThreshold}.
	\begin{figure}[p]
		\centering
		\includegraphics[scale=0.8]{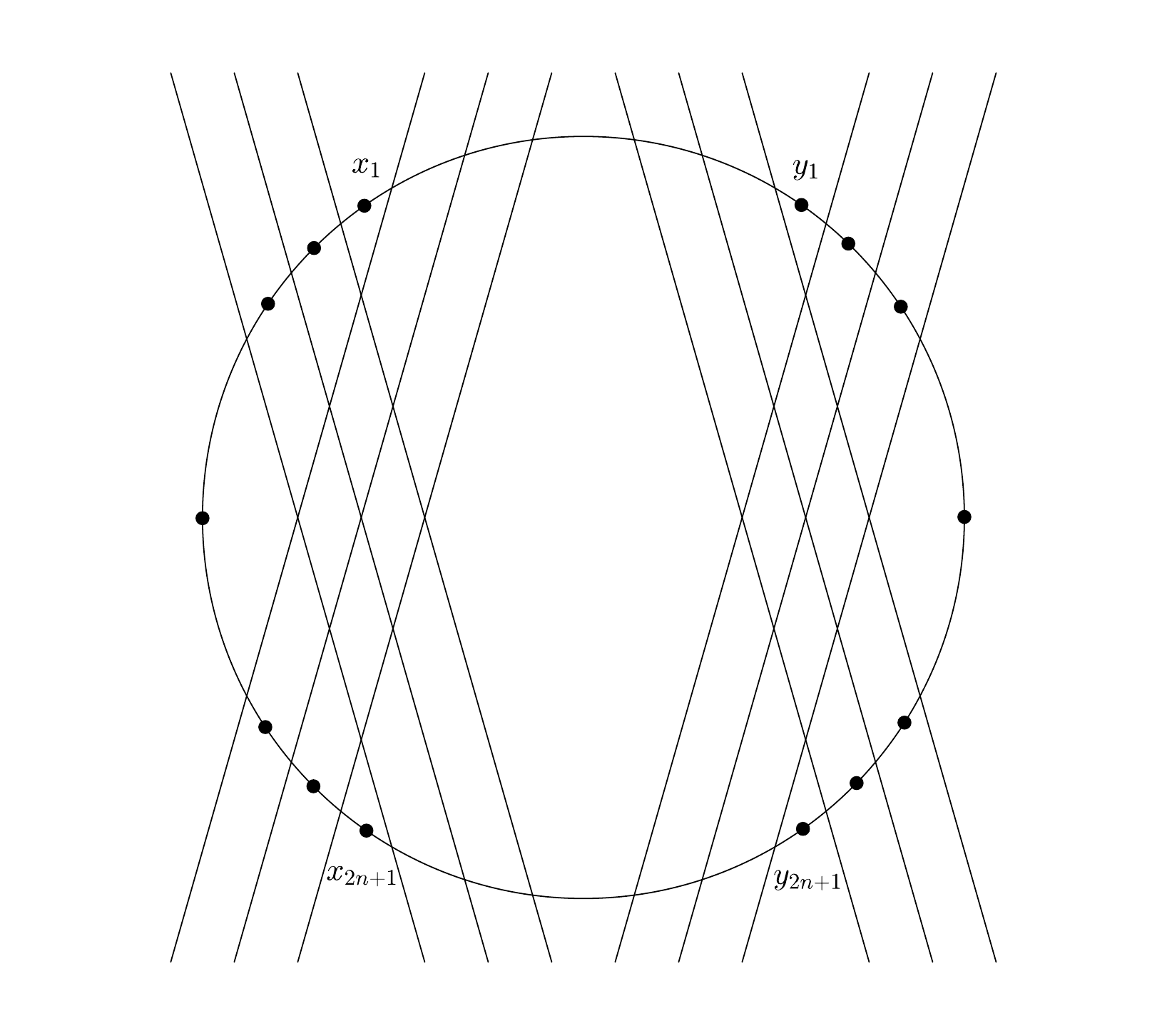}
		\caption{Example showing tightness of the factor 2 in the hoppability bound.}
		\label{ConnectThreshold}
	\end{figure}
	In this graph, every vertex denotes a taxon and every line denotes a split of weight 1. In this split network, $\Delta = 1$, $\Omega = 3$, and the smallest distance between those taxa on the left and those taxa on the right is 6. We can generalize this network by increasing the number of splits in the 4 subsets of parallel splits from $3$ to $n$ each. Then, $\Delta$ is still 1, while $\Omega$ is $n$, and the smallest distance between the taxa on the left and the taxa on the right becomes $2n$.	
	%	Furthermore, if we add four more splits to the network:
	%	\begin{eqnarray*}
	%		\{x_1, \Taxa - \{x_1\}\}, \{x_{2n+1}, \Taxa - \{x_{2n+1}\}\}, \{y_1, \Taxa - \{y_1\}\}, \{y_{2n+1}, \Taxa - \{y_{2n+1}\}\}
	%	\end{eqnarray*}
	%	then $\Delta = 2$, $\Omega = n$, and the smallest distance between taxa on the left and that on the right is $2n+2$. This shows the tightness of $\Delta+2\Omega$.
\end{example}

\paragraph{Bounding the distance between taxa separated by a split.} The following bound will be useful.
\begin{lemma}[A distance bound]
\label{OnlyOneSideLemma}
Let $\Network$ be a split
network with chord depth $\Delta$ and
maximum incompatibility $\Omega$. 
Suppose $x,y, z, z'\in \Taxa$. If there exists a split $S\in\Splits$ that separates $\{z, x\}$ and $\{z', y\}$, then $d(z,x)+d(z,y)\leq 2d(z,z')+d(x,y)+2\Omega$ or
$$
d(z,z') \geq \frac{1}{2}\left[
d(z,x)+d(z,y) -d(x,y) 
\right] - \Omega.
$$
\end{lemma}

\subsection{Proof of Theorem~\ref{NetworkReconstructionLabel}}

\begin{proof}[Proof of Theorem~\ref{NetworkReconstructionLabel}]
The main theorem now follows from
Propositions~\ref{prop:mini-correct},~\ref{prop:ext-correct},~\ref{prop:exhaust} and~\ref{prop:weight-estimates},
which are proved %in the appendix.
in Section~\ref{sec:appendix}.

Regarding the computational complexity of the algorithm, note that there are $O(n^2)$ pairs of $x,y$ which satisfy $\hat{d}(x,y) = \ellipseC$. For each $B(x,y)$, the split decomposition method takes $O(n^6)$ time and generates at most $O(n^2)$ of $\hat{d}$-split~\cite{SplitDecomposition}. For every split in $\SmallSplits$, the while loop of Bipartition Extension takes $O(n^2)$ time using DFS. Thus, the running time of the algorithm is
\begin{eqnarray*}
	O(n^2)\cdot(O(n^6) + O(n^2)\cdot O(n^2)) = O(n^8).
\end{eqnarray*}
\end{proof}

\paragraph{More on the computational complexity.} We did not attempt to optimize the computational complexity of our reconstruction algorithm. In fact, the split decomposition method can be replaced with the much faster Neighbor-Net algorithm~\cite{NeighborNet}, which runs in time $O(n^3)$ (with a modified weight estimation step~\cite{NNworks}). Our analysis can then be adapted using results in~\cite{NNworks}. Further speed-up can be obtained by reconstructing small regions around a {\em single} taxon (rather than $2$), at the expense of a slightly larger radius.

\clearpage

\bibliographystyle{alpha}
\bibliography{bibtex}

\begin{thebibliography}{RMWW04}

\bibitem[BD92]{SplitDecomposition}
H.~J. Bandelt and A.~W.~M. Dress.
\newblock A canonical decomposition theory for metrics on a finite set.
\newblock {\em Advances in mathematics}, 92(1):47--105, 1992.

\bibitem[BM04]{NeighborNet}
D.~Bryant and V.~Moulton.
\newblock Neighbor-net: an agglomerative method for he construction of
  phylogenetic networks.
\newblock {\em Molecular Biology and Evolution}, 21(2):255--265, 2004.

\bibitem[Bry05]{GroupBasedModel}
D.~Bryant.
\newblock Extending tree models to split networks.
\newblock {\em Algebraic Statistics for Computational Biology (L Pachter and B
  Sturmfels, editors), Cambridge University Press}, pages 297--310, 2005.

\bibitem[Bun71]{Compatibility}
P.~Buneman.
\newblock The recovery of trees from measures of dissimilarity. in {D.G.
  Kendall and P. Tautu}, editors.
\newblock {\em Mathematics in the Archaeological and Historical Sciences},
  pages 387--395, 1971.

\bibitem[CGG02]{CrGoGo:02}
M.~Cryan, L.~A. Goldberg, and P.~W. Goldberg.
\newblock Evolutionary trees can be learned in polynomial time.
\newblock {\em SIAM J. Comput.}, 31(2):375--397, 2002.
\newblock short version, Proceedings of the 39th Annual Symposium on
  Foundations of Computer Science (FOCS 98), pages 436-445, 1998.

\bibitem[DMR11]{TreeCase}
C.~Daskalakis, E.~Mossel, and S.~Roch.
\newblock Phylogenies without branch bounds; contracting the short, pruning the
  deep.
\newblock {\em SIAM Journal of Discrete Math}, 25(2):872--893, 2011.

\bibitem[ESSW99a]{ErStSzWa:99a}
P.~L. Erd\"{o}s, M.~A. Steel, L.~A. Sz\'{e}kely, and T.~A. Warnow.
\newblock A few logs suffice to build (almost) all trees (part 1).
\newblock {\em Random Struct. Algor.}, 14(2):153--184, 1999.

\bibitem[ESSW99b]{ErStSzWa:99b}
P.~L. Erd\"{o}s, M.~A. Steel, L.~A. Sz\'{e}kely, and T.~A. Warnow.
\newblock A few logs suffice to build (almost) all trees (part 2).
\newblock {\em Theor. Comput. Sci.}, 221:77--118, 1999.

\bibitem[Fel04]{Felsenstein:04}
J.~Felsenstein.
\newblock {\em Inferring Phylogenies}.
\newblock Sinauer, Sunderland, MA, 2004.

\bibitem[GMS12]{RSA:RSA20372}
Ilan Gronau, Shlomo Moran, and Sagi Snir.
\newblock Fast and reliable reconstruction of phylogenetic trees with
  indistinguishable edges.
\newblock {\em Random Structures and Algorithms}, 40(3):350--384, 2012.

\bibitem[HB06]{Huson01022006}
Daniel~H. Huson and David Bryant.
\newblock Application of phylogenetic networks in evolutionary studies.
\newblock {\em Molecular Biology and Evolution}, 23(2):254--267, 2006.

\bibitem[HNW99]{DCM}
Daniel~H. Huson, Scott~M. Nettles, and Tandy~J. Warnow.
\newblock Disk-covering, a fast-converging method for phylogenetic tree
  reconstruction.
\newblock {\em Journal of Computational Biology}, 6(3-4):369--386, 2016/09/14
  1999.

\bibitem[HRS10]{Book}
D.~H. Huson, R.~Rupp, and C.~Scornavacca.
\newblock {\em Phylogenetic Networks: Concepts, Algorithms and Applications}.
\newblock Cambridge, 2010.

\bibitem[JC69]{JukesCantor}
T.~H. Jukes and C.~R. Cantor.
\newblock Evolution of protein molecules.
\newblock {\em Mammalian Protein Metabolism, Academic Press, New York.}, pages
  21--132, 1969.

\bibitem[KZZ03]{03KZZ}
V.~King, L.~Zhang, and Y.~Zhou.
\newblock On the complexity of distance-based evolutionary tree reconstruction.
\newblock {\em in Proceedings of the 14th Annual ACM-SIAM Symposium on Discrete
  Algorithms, SIAM, Philadelphia 2003}, pages 444--453, 2003.

\bibitem[LC06]{LaceyChang:06}
Michelle~R. Lacey and Joseph~T. Chang.
\newblock A signal-to-noise analysis of phylogeny estimation by
  neighbor-joining: insufficiency of polynomial length sequences.
\newblock {\em Math. Biosci.}, 199(2):188--215, 2006.

\bibitem[LP11]{NNworks}
D.~Levy and L.~Pachter.
\newblock The neighbor-net algorithm.
\newblock {\em Advances in Applied Mathematics}, 47(2):240--258, 2011.

\bibitem[Mos07]{Mossel:07}
E.~Mossel.
\newblock Distorted metrics on trees and phylogenetic forests.
\newblock {\em IEEE/ACM Trans. Comput. Bio. Bioinform.}, 4(1):108--116, 2007.

\bibitem[MR06]{MosselRoch:06}
Elchanan Mossel and S{\'e}bastien Roch.
\newblock Learning nonsingular phylogenies and hidden {M}arkov models.
\newblock {\em Ann. Appl. Probab.}, 16(2):583--614, 2006.

\bibitem[RMWW04]{DCM3}
Usman~W. Roshan, Bernard M.~E. Moret, Tandy Warnow, and Tiffani~L. Williams.
\newblock {Rec-I-DCM3}: A fast algorithmic technique for reconstructing large
  phylogenetic trees.
\newblock {\em Computational Systems Bioinformatics Conference, International
  IEEE Computer Society}, 0:98--109, 2004.

\bibitem[SS03]{SempleSteel:03}
C.~Semple and M.~Steel.
\newblock {\em Phylogenetics}, volume~22 of {\em Mathematics and its
  Applications series}.
\newblock Oxford University Press, 2003.

\bibitem[Ste94]{Steel:94}
M.~Steel.
\newblock Recovering a tree from the leaf colourations it generates under a
  {M}arkov model.
\newblock {\em Appl. Math. Lett.}, 7(2):19--23, 1994.

\bibitem[Ste16]{Steel:16}
Mike Steel.
\newblock {\em Phylogeny---discrete and random processes in evolution},
  volume~89 of {\em CBMS-NSF Regional Conference Series in Applied
  Mathematics}.
\newblock Society for Industrial and Applied Mathematics (SIAM), Philadelphia,
  PA, 2016.

\bibitem[War]{Warnow:u}
Tandy Warnow.
\newblock \textit{Computational phylogenetics: An introduction to designing
  methods for phylogeny estimation}.
\newblock To be published by Cambridge University Press , 2017.

\end{thebibliography}

\clearpage

\appendix

\section{Proofs}
\label{sec:appendix}

\subsection{Key lemmas}
\label{sec:key-lemmas}

\begin{proof}[Proof of Lemma~\ref{WitnessOfDepth}]
	Let $x,y$ be a pair of $S$-witnesses, which are guaranteed to exist by definition of the chord depth. Then
	\begin{equation*}
	d(x,y) = d(x,y;\Comp(S)) + d(x,y;\Inco(S)) 
	\leq \Delta + \sum_{S' \in \Inco(S)} w(S') \leq \Delta + \Omega,
	\end{equation*}
	where the equality comes from the fact that
	$\Splits|_{x,y} = [\Splits|_{x,y}\cap \Comp(S)] \cup [\Splits|_{x,y}\cap \Inco(S)]$ and $\Comp(S) \cap \Inco(S) = \emptyset$.
\end{proof}

\begin{proof}[Proof of Lemma~\ref{lemma:hoppability}]
	We first provide some intuition for the proof
	by considering the tree case. Assume $\Omega = 0$
	and let $x,y \in \Taxa$. Let $T = (V,E)$ be a binary tree corresponding
	to $\Net$, where recall that $\Taxa$ is the set of leaves
	of $T$. To connect $x$ and $y$ through short hops, it
	is natural to consider the unique path $\mathscr{P} = (x=u_0,u_1,\ldots,u_{r+1}=y)$
	connecting $x$ and $y$ on $T$, where $(u_i,u_{i+1}) \in E$ for $i=0,\ldots,r$. Note however that, except for the endpoints, the $u_i$s are not in $\Taxa$. 
	However we can associate a ``close-by leaf'' to each $u_i$.
	For $i=1,\ldots,r$, let
	$v_i$ be the neighbor of $u_i$ not on $\mathscr{P}$
	and let $\Taxa_i$ be the subset of $\Taxa$ reachable
	from $v_i$ without visiting $u_i$. Let $w_i \in \Taxa_i$ be the closest leaf to $u_i$ in $\Taxa_i$. Then, by hopping along $x=w_0, w_1,\ldots,w_{r+1} =y$, one can easily show that
	$\Net$ is $3\Delta$-hoppable. (As a side result of our more delicate, general analysis below, we show that $\Net$ is in reality $\Delta$-hoppable.)
	
	For a general split network, 
	we replace the path $\mathscr{P}$ above with what we refer to 
	as a compatible chain. We first introduce some notation.
	Fix $x,y \in \Taxa$.
	For two splits $S, S' \in \Splits$, we write
	$S \prec_x S'$ if $S_x \subset S'_x$. Note that
	$\prec_x$ is transitive. Also, observe the following. 
	\begin{lemma}[Separation and compatibility imply total order]
		\label{lemma:comp-prec}
		Let $S, S' \in \Splits|_{x,y}$ with $S \neq S'$.
		Then $S \sim S'$ if and only if either
		$S \prec_x S'$ or $S' \prec_x S$.
	\end{lemma}
	\begin{proof}
		Assume $S \sim S'$. Then one of the following sets is empty: $S_x \cap S_x'$, $S_y \cap S'_y$, $S_x \cap S_y'$ or $S_y \cap S_x'$. The first two sets cannot be empty because of their inclusion of $x$ or $y$. Suppose $S_x \cap S'_y = \emptyset$. Then 
		$$
		S_x = (S_x\cap S'_x) \cup (S_x\cap S'_y) = S_x \cap S_x', 
		$$
		where we used that by definition $S_x \cap S_y = \emptyset$ and $S'_x \cap S_y' = \emptyset$.
		Thus $S_x \subseteq S_x'$. Because $S \neq S'$, we get $S \prec_x S'$. And similarly for the other case.
		
		Conversely, assume that $S \prec_x S'$. Then $S_x \subset S_x'$ which implies that $S_x \cap (S_x')^c = \emptyset$. Now note that $(S_x')^c = S_y'$, which implies compatibility. And similarly for the other case.
	\end{proof}
	\noindent A {\bf compatible chain separating $x$ and $y$} is a collection
	of distinct splits $S^{(1)}, \ldots, S^{(\ell)}$ such that:
	for all $i$, $S^{(i)} \in \Splits|_{x,y}$ and, for all $i,j$, 
	$S^{(i)} \sim S^{(j)}$. By Lemma~\ref{lemma:comp-prec}, we can assume without loss of generality that
	$S^{(1)} \prec_x \cdots \prec_x S^{(\ell)}$. Assume
	further that $S^{(1)}, \ldots, S^{(\ell)}$ is a {\bf maximal}
	such chain, that is, for any $S' \in \Splits|_{x,y}$ not in the chain 
	there is an $i$ such that $S'\nsim S^{(i)}$. For all $i$, let
	$u_x^{(i)}, u_y^{(i)} \in \Taxa$ be a pair of 
	$S^{(i)}$-witnesses, that is,  
	\begin{equation}
	\label{eq:consec-delta}
	d\left(u_x^{(i)}, u_y^{(i)};\Comp\left(S^{(i)}\right)\right) \leq \Delta,
	\end{equation}
	and $u_x^{(i)} \in S^{(i)}_x$, $u_y^{(i)} \in S^{(i)}_y$.
	The key claim in our proof of Lemma~\ref{lemma:hoppability} is the following.
	\begin{lemma}[Chain hoppability]
		\label{lemma:chain-hop}
		Let $S^{(1)}, \ldots, S^{(\ell)}$ be a 
		maximal compatible chain separating $x$ and $y$
		and, for $i=1,\ldots,\ell$, let $u_x^{(i)}, u_y^{(i)} \in \Taxa$ be as above.
		\begin{enumerate}
			\item[(a)] For all $i = 1,\ldots,\ell-1$, we have
			either
			\begin{equation}
			\label{eq:consec-case1}	
			d\left(u_y^{(i)},u_y^{(i+1)}\right)\leq \kappa,
			\end{equation}
			or
			\begin{equation}
			\label{eq:consec-case2}	
			d\left(u_x^{(i)},u_x^{(i+1)}\right)\leq \kappa.
			\end{equation}
			
			\item[(b)] We can always choose
			$u^{(1)}_x = x$ and $u^{(\ell)}_y = y$.
			
		\end{enumerate}
	\end{lemma}
	\noindent Before proving Lemma~\ref{lemma:chain-hop},
	we show that it implies Lemma~\ref{lemma:hoppability}.
	We need to construct a sequence of $\kappa$-hops
	between $x$ and $y$.
	Let $z_0 = x$, $z_{\ell} = y$ and, for $i=1,\ldots,\ell-1$, $z_i = u_y^{(i)}$. 
	For all $i$ such that Equation~\eqref{eq:consec-case1} holds, the pair $(z_i,z_{i+1})$ is indeed a $\kappa$-hop. 
	That may not be the case however
	for those $i$ such that Equation~\eqref{eq:consec-case2} holds.
	Instead, in that case, we ``backtrack''
	to $u^{(i)}_x$, move on to $u_x^{(i+1)}$ and then on to
	$u_y^{(i+1)}$. Indeed, by~\eqref{eq:consec-delta}, it holds
	that 
	$$
	d\left(u_x^{(i)}, u_y^{(i)}\right) 
	= d\left(u_x^{(i)}, u_y^{(i)};\Comp\left(S^{(i)}\right)\right)
	+ d\left(u_x^{(i)}, u_y^{(i)};\Inco\left(S^{(i)}\right)\right) \leq \Delta + \Omega \leq \kappa,
	$$
	and the same holds for $i+1$.
	Formally, for each $i=1,\ldots,\ell-1$
	: if Case~\eqref{eq:consec-case1} holds, we
	let $\dot{z}_i = u_y^{(i)}$ and $\ddot{z}_i = u_y^{(i+1)}$;
	if Case~\eqref{eq:consec-case2} holds, we
	let $\dot{z}_i = u_x^{(i)}$ and $\ddot{z}_i = u_x^{(i+1)}$.
	Then the sequence 
	$$
	z_0,z_1, \dot{z}_1, \ddot{z}_1, \ldots, z_{\ell-1},\dot{z}_{\ell-1},\ddot{z}_{\ell-1},z_\ell.
	$$
	establishes $\kappa$-hoppability.
	
	It remains to prove Lemma~\ref{lemma:chain-hop}.
	\begin{proof}[Proof of Lemma~\ref{lemma:chain-hop}]
		Fix $i$. We seek to upper bound $d(u_x^{(i)},u_x^{(i+1)}) + d(u_y^{(i)},u_y^{(i+1)})$.
		We decompose the sum into contributions
		compatible with $S^{(i)}$ and $S^{(i+1)}$
		and contributions
		incompatible with one of them.
		The latter is straighforward to bound. 
		\begin{claim}[Incompatible contributions: bound]
			\label{claim:1}
			We have
			\begin{equation}
			\label{eq:hop-inco-bound1}
			d\left(u_x^{(i)},u_x^{(i+1)};\Inco\left(S^{(i)}\right)
			\cup \Inco\left(S^{(i+1)}\right)\right)
			\leq 2\Omega,
			\end{equation}
			and 
			\begin{equation}
			\label{eq:hop-inco-bound2}
			d\left(u_y^{(i)},u_y^{(i+1)};\Inco\left(S^{(i)}\right)
			\cup \Inco\left(S^{(i+1)}\right)\right)\leq 2\Omega.	
			\end{equation}
		\end{claim}
		\begin{proof}
			By definition of $\Omega$,
			\begin{eqnarray*}
				d\left(u_x^{(i)},u_x^{(i+1)};\Inco\left(S^{(i)}\right)
				\cup \Inco\left(S^{(i+1)}\right)\right)
				&\leq& d\left(u_x^{(i)},u_x^{(i+1)};\Inco\left(S^{(i)}\right)
				\right)
				+
				d\left(u_x^{(i)},u_x^{(i+1)};\Inco\left(S^{(i+1)}\right)\right)\\
				&\leq&
				\sum_{S \in \Inco(S^{(i)})} w(S)
				+
				\sum_{S \in \Inco(S^{(i+1)})} w(S)\\
				&\leq& 2 \Omega,
			\end{eqnarray*}
			and similarly for the other inequality.
		\end{proof}
		
		To bound the compatible contributions, we further subdivide them into whether or not they separate $x$ and $y$.  We let
		\begin{equation}
		\label{eq:hop-comp-contr1}
		\sep_x\left(S^{(i)}\right) = \left\{S \in \Comp\left(S^{(i)}\right) \cap \Splits|_{x,y}\,:\, S=S^{(i)}\text{ or }
		S \prec_x S^{(i)} \right\}, 		
		\end{equation}
		be the set of splits compatible with $S^{(i)}$ separating $x$ and $y$ on
		the ``$x$ side of $S^{(i)}$''
		and, similarly, we let
		\begin{equation}
		\label{eq:hop-comp-contr2}
		\sep_y\left(S^{(i+1)}\right) = \left\{S \in \Comp\left(S^{(i+1)}\right)  \cap \Splits|_{x,y}\,:\, S=S^{(i+1)}\text{ or }
		S^{(i+1)} \prec_x S\right\}. 	
		\end{equation}
		By maximality of the chain and Lemma~\ref{lemma:comp-prec}, we have that
		\begin{equation}
		\label{eq:hop-comp-contr12}
		\sep_x\left(S^{(i)}\right)
		\sqcup \sep_y\left(S^{(i+1)}\right)
		= \Comp\left(S^{(i)},S^{(i+1)}\right)  \cap \Splits|_{x,y},
		\end{equation}
		where $\sqcup$ indicates that the sets in the union are disjoint.
		\begin{claim}[Compatible contributions: separating]
			\label{claim:2}
			We have
			\begin{eqnarray}
			&&d\left(u_x^{(i)},u_x^{(i+1)};\sep_x\left(S^{(i)}\right)\right)
			+
			d\left(u_y^{(i)},u_y^{(i+1)};\sep_y\left(S^{(i+1)}\right)\right)\nonumber\\
			&&\qquad \leq 
			d\left(u_x^{(i)},u_y^{(i)}; \Comp\left(S^{(i)},S^{(i+1)}\right)  \cap \Splits|_{x,y}\right)
			+
			d\left(u_x^{(i+1)},u_y^{(i+1)}; \Comp\left(S^{(i)},S^{(i+1)}\right)  \cap \Splits|_{x,y}\right).			\label{eq:hop-comp-bound}
			\end{eqnarray}	 	
		\end{claim}
		\begin{proof}
			We relate the distances on the l.h.s.~of~\eqref{eq:hop-comp-bound} to the distance
			between $u_x^{(i)}$ and $u_y^{(i)}$ and
			the distance between $u_x^{(i+1)}$ and $u_y^{(i+1)}$.
			We argue about the first term. We claim
			that 
			\begin{equation}
			\label{eq:hop-comp-inclusion1}
			\Splits|_{u_x^{(i)},u_x^{(i+1)}}\cap \sep_x\left(S^{(i)}\right)
			\subseteq \left[\Splits|_{u_x^{(i)},u_y^{(i)}}
			\cap \sep_x\left(S^{(i)}\right)\right]
			\cup 
			\left[
			\Splits|_{u_x^{(i+1)},u_y^{(i+1)}}
			\cap \sep_x\left(S^{(i)}\right)\right].
			\end{equation}
			In words, splits in $\sep_x\left(S^{(i)}\right)$
			that are separating $u_x^{(i)},u_x^{(i+1)}$
			are also separating either
			$u_x^{(i)},u_y^{(i)}$ or $u_x^{(i+1)},u_y^{(i+1)}$.
			Indeed, let $S \in \Splits|_{u_x^{(i)},u_x^{(i+1)}}\cap \sep_x\left(S^{(i)}\right)$.
			We observe first that by definition
			\begin{equation}
			\label{eq:hop-comp-uys}
			u_y^{(i)} \in S_y^{(i)}
			\qquad \text{and} \qquad  u_y^{(i+1)} \in S_y^{(i+1)} \subseteq S_y^{(i)}.
			\end{equation}
			Because $S$ separates $u_x^{(i)},u_x^{(i+1)}$,
			it follows that $y$ may be on either side of $S$.
			We consider the two cases separately:
			\begin{enumerate}
				\item Assume that $y \in S_{u_x^{(i)}}$.
				Because
				$S \in \sep_x\left(S^{(i)}\right)$, we must have
				that $S_{u_x^{(i)}} \supseteq S^{(i)}_y \owns u_y^{(i+1)}$, where we used~\eqref{eq:hop-comp-uys}.
				Hence $S \in \Splits|_{u_x^{(i+1)},u_y^{(i+1)}}\cap \sep_x\left(S^{(i)}\right)$.
				
				\item Similarly, in the other case, $y \in S_{u_x^{(i+1)}}$ implies that
				$S \in \Splits|_{u_x^{(i)},u_y^{(i)}}\cap \sep_x\left(S^{(i)}\right)$.
				
			\end{enumerate}
			That proves~\eqref{eq:hop-comp-inclusion1}.
			
			Similarly, we can show that
			\begin{equation}
			\label{eq:hop-comp-inclusion2}
			\Splits|_{u_y^{(i)},u_y^{(i+1)}}\cap \sep_y\left(S^{(i+1)}\right)
			\subseteq \left[
			\Splits|_{u_x^{(i)},u_y^{(i)}}\cap \sep_y\left(S^{(i+1)}\right)\right]
			\cup 
			\left[
			\Splits|_{u_x^{(i+1)},u_y^{(i+1)}}\cap \sep_y\left(S^{(i+1)}\right)\right].
			\end{equation}
			Combining~\eqref{eq:hop-comp-inclusion1} 
			and~\eqref{eq:hop-comp-inclusion2}, and
			using~\eqref{eq:consec-delta},
			\begin{eqnarray*}
				&&d\left(u_x^{(i)},u_x^{(i+1)};\sep_x\left(S^{(i)}\right)\right)
				+
				d\left(u_y^{(i)},u_y^{(i+1)};\sep_y\left(S^{(i+1)}\right)\right)\\
				&&\qquad\qquad \leq 
				d\left(u_x^{(i)},u_y^{(i)}; \sep_x\left(S^{(i)}\right)\right)
				+
				d\left(u_x^{(i+1)},u_y^{(i+1)}; \sep_x\left(S^{(i)}\right)\right)\\
				&&\qquad\qquad\qquad+
				d\left(u_x^{(i)},u_y^{(i)}; \sep_y\left(S^{(i+1)}\right)\right)
				+
				d\left(u_x^{(i+1)},u_y^{(i+1)}; \sep_y\left(S^{(i+1)}\right)\right)\\
				&&\qquad\qquad =
				\left[
				d\left(u_x^{(i)},u_y^{(i)}; \sep_x\left(S^{(i)}\right)\right)
				+
				d\left(u_x^{(i)},u_y^{(i)}; \sep_y\left(S^{(i+1)}\right)\right)
				\right]\\
				&&\qquad\qquad\qquad+
				\left[
				d\left(u_x^{(i+1)},u_y^{(i+1)}; \sep_x\left(S^{(i)}\right)\right)
				+
				d\left(u_x^{(i+1)},u_y^{(i+1)}; \sep_y\left(S^{(i+1)}\right)\right)
				\right]\\
				&&\qquad\qquad =
				d\left(u_x^{(i)},u_y^{(i)}; \Comp\left(S^{(i)},S^{(i+1)}\right)  \cap \Splits|_{x,y}\right)
				+
				d\left(u_x^{(i+1)},u_y^{(i+1)}; \Comp\left(S^{(i)},S^{(i+1)}\right)  \cap \Splits|_{x,y}\right),
			\end{eqnarray*} 
			where the last equality follows from~\eqref{eq:hop-comp-contr12}.
			That proves~\eqref{eq:hop-comp-bound}.
		\end{proof}
		
		We consider now the non-separating contributions.
		Let $S \notin \Splits|_{x,y}$. We let $S = \{S_0,S_{x,y}\}$
		where $x,y \in S_{x,y}$.
		We define
		\begin{equation}
		\label{eq:hop-comp-contr3}
		\nsep_x\left(S^{(i)}\right) = \left\{S \in \Comp\left(S^{(i)}\right) - \Splits|_{x,y}\,:\, S_0 \subset S^{(i)}_x \right\}, 		
		\end{equation}
		to be the set of splits compatible with $S^{(i)}$ not separating $x$ and $y$ on
		the ``$x$ side of $S^{(i)}$''
		and, similarly, we let
		\begin{equation}
		\label{eq:hop-comp-contr4}
		\nsep_y\left(S^{(i+1)}\right) = \left\{S \in \Comp\left(S^{(i+1)}\right)  - \Splits|_{x,y}\,:\, S_0 \subset S_y^{(i+1)}\right\}, 	
		\end{equation}
		and
		\begin{equation}
		\label{eq:hop-comp-contr5}
		\nsep_{x,y}\left(S^{(i)},S^{(i+1)}\right) = \left\{S \in \Comp\left(S^{(i)},S^{(i+1)}\right)  - \Splits|_{x,y}\,:\, S_0 \subset S_y^{(i)} \cap S_x^{(i+1)}\right\}. 	
		\end{equation}
		Using that $S^{(i)} \prec_x S^{(i+1)}$, it follows that
		\begin{equation}
		\label{eq:hop-comp-contr345}
		\nsep_x\left(S^{(i)}\right)
		\sqcup
		\nsep_{x,y}\left(S^{(i)},S^{(i+1)}\right)
		\sqcup \nsep_y\left(S^{(i+1)}\right)
		= \Comp\left(S^{(i)},S^{(i+1)}\right)  - \Splits|_{x,y}.
		\end{equation}
		
		\begin{claim}[Compatible contributions: non-separating]
			\label{claim:3}
			We have
			\begin{eqnarray}
			d\left(u_x^{(i)},u_x^{(i+1)};\nsep_x\left(S^{(i)}\right)\right)
			\leq 
			d\left(u_x^{(i)},u_y^{(i)};\nsep_x\left(S^{(i)}\right)\right) +
			d\left(u_x^{(i+1)},u_y^{(i+1)};\nsep_x\left(S^{(i)}\right)\right)
			,\label{eq:hop-comp-nonsep1}
			\end{eqnarray}
			\begin{eqnarray}
			d\left(u_x^{(i)},u_x^{(i+1)};\nsep_{x,y}\left(S^{(i)},S^{(i+1)}\right)\right)
			\leq 
			d\left(u_x^{(i+1)},u_y^{(i+1)};\nsep_{x,y}\left(S^{(i)},S^{(i+1)}\right)\right)
			,\label{eq:hop-comp-nonsep2}
			\end{eqnarray}
			and
			\begin{eqnarray}
			d\left(u_x^{(i)},u_x^{(i+1)};\nsep_y\left(S^{(i+1)}\right)\right)
			=0.\label{eq:hop-comp-nonsep3}
			\end{eqnarray}
			And similarly for the pair $(u_y^{(i)},u_y^{(i+1)})$ with the roles of
			$x$ and $y$ and the roles of $i$ and $i+1$ interchanged respectively. 
		\end{claim}
		\begin{proof}
			For~\eqref{eq:hop-comp-nonsep1}, let
			$S \in \Splits|_{u_x^{(i)},u_x^{(i+1)}} \cap \nsep_x\left(S^{(i)}\right)$. If $u_x^{(i)} \in S_0 \subset S_x^{(i)}$, then we have that
			$u_y^{(i)} \in S_y^{(i)} \subset S_{x,y}$ so that
			$S \in \Splits|_{u_x^{(i)},u_y^{(i)}}$. Similarly,
			if $u_x^{(i+1)} \in S_0$, we have $S \in \Splits|_{u_x^{(i+1)},u_y^{(i+1)}}$.
			
			For~\eqref{eq:hop-comp-nonsep2}, let
			$S \in \Splits|_{u_x^{(i)},u_x^{(i+1)}} \cap \nsep_{x,y}\left(S^{(i)},S^{(i+1)}\right)$. By definition of $\nsep_{x,y}\left(S^{(i)},S^{(i+1)}\right)$, we have that $u_x^{(i)} \notin S_0 \subset S_y^{(i)}$. Since
			$S \in \Splits|_{u_x^{(i)},u_x^{(i+1)}}$, we must have that  $u_x^{(i+1)} \in S_0 \subset S_x^{(i+1)}$. Hence $u_y^{(i+1)} \in S_y^{(i+1)} \subset S_{x,y}$ and $S \in \Splits|_{u_x^{(i+1)},u_y^{(i+1)}}$.
			
			Finally let $S \in \nsep_y\left(S^{(i+1)}\right)$. We have that $S_0 \subset S_y^{(i+1)}$ so neither $u_x^{(i)}$ nor $u_x^{(i+1)}$ is in $S_0$ and therefore $S \notin \Splits|_{u_x^{(i)},u_x^{(i+1)}}$. 
		\end{proof}
		
		By Claims~\ref{claim:1},~\ref{claim:2}, and~\ref{claim:3} together with~\eqref{eq:hop-comp-contr12} and~\eqref{eq:hop-comp-contr345}, we get that
		\begin{eqnarray*}
			&&d(u_x^{(i)},u_x^{(i+1)}) + d(u_y^{(i)},u_y^{(i+1)})\\
			&& \qquad\qquad = d\left(u_x^{(i)},u_x^{(i+1)};\Comp\left(S^{(i)},S^{(i+1)}\right)\right) + d\left(u_y^{(i)},u_y^{(i+1)};\Comp\left(S^{(i)},S^{(i+1)}\right)\right)\\
			&&\qquad\qquad\qquad + d\left(u_x^{(i)},u_x^{(i+1)};\Inco\left(S^{(i)},S^{(i+1)}\right)\right) + d\left(u_y^{(i)},u_y^{(i+1)};\Inco\left(S^{(i)},S^{(i+1)}\right)\right)\\
			&& \qquad\qquad \leq d\left(u_x^{(i)},u_y^{(i)};\Comp\left(S^{(i)},S^{(i+1)}\right)\right) + d\left(u_x^{(i+1)},u_y^{(i+1)};\Comp\left(S^{(i)},S^{(i+1)}\right)\right)\\
			&&\qquad\qquad\qquad + d\left(u_x^{(i)},u_x^{(i+1)};\Inco\left(S^{(i)},S^{(i+1)}\right)\right) + d\left(u_y^{(i)},u_y^{(i+1)};\Inco\left(S^{(i)},S^{(i+1)}\right)\right)\\
			&&\qquad\qquad  \leq 2\Delta + 4\Omega = 2 \kappa.
		\end{eqnarray*}
		That implies (a) in Lemma~\ref{lemma:chain-hop}.
		
		For (b), by the maximality of the chain and Lemma~\ref{lemma:comp-prec}, we have that
		$\sep_x(S^{(1)}) = \{S^{(1)}\}$. Hence,
		$\Splits|_{x,u_y^{(1)}} \cap \sep_x(S^{(1)}) \subseteq \Splits|_{u_x^{(1)},u_y^{(1)}} \cap \sep_x(S^{(1)}) $. 
		Also, because $u_y^{(1)} \in S_y^{(1)}$, we have that $\Splits|_{x,u_y^{(1)}} \cap \nsep_x(S^{(1)}) = \emptyset$ so, by default, $\Splits|_{x,u_y^{(1)}} \cap \nsep_x(S^{(1)}) \subseteq \Splits|_{u_x^{(1)},u_y^{(1)}} \cap \nsep_x(S^{(1)}) $. Finally
		$\Splits|_{x,u_y^{(1)}}\cap(\sep_y(S^{(1)})\cup \nsep_y(S^{(1)})-\{S^{(1)}\}) = 	\Splits|_{u_x^{(1)},u_y^{(1)}}\cap(\sep_y(S^{(1)})\cup \nsep_y(S^{(1)})-\{S^{(1)}\})$ because $x, u_x^{(1)} \in S_x^{(1)}$.  
		So
		$$
		d\left(u_x^{(1)},u_y^{(1)};\Comp\left(S^{(1)}\right)\right)\geq d\left(x,u_y^{(1)};\Comp\left(S^{(1)}\right)\right).
		$$
		And similarly for the claim about $y$.
	\end{proof}
	That concludes the proof of Lemma~\ref{lemma:hoppability}.
\end{proof}

\noindent Before proving Lemma~\ref{OnlyOneSideLemma}, we introduce some notation.
Suppose $x_1, x_2, y_1, y_2 \in \Taxa$. Let $w_{x_1}$ be the total weight that separates $\{x_1\}$ from the three other points together, with similar definitions for $w_{y_1}$, $w_{x_2}$, and $w_{y_2}$. Moreover, let $w_1$ be the total weight that separates $\{x_1, y_1\}$ and $\{x_2, y_2\}$, $w_2$ be the total weight that separates $\{x_1, x_2\}$ and $\{y_1, y_2\}$, and $w_3$ be the total weight that separate $\{x_1, y_2\}$ and $\{x_2, y_1\}$, as shown \Cref{DistanceGraph}.
\begin{figure}[ht]
	\centering
	\includegraphics[scale=0.2]{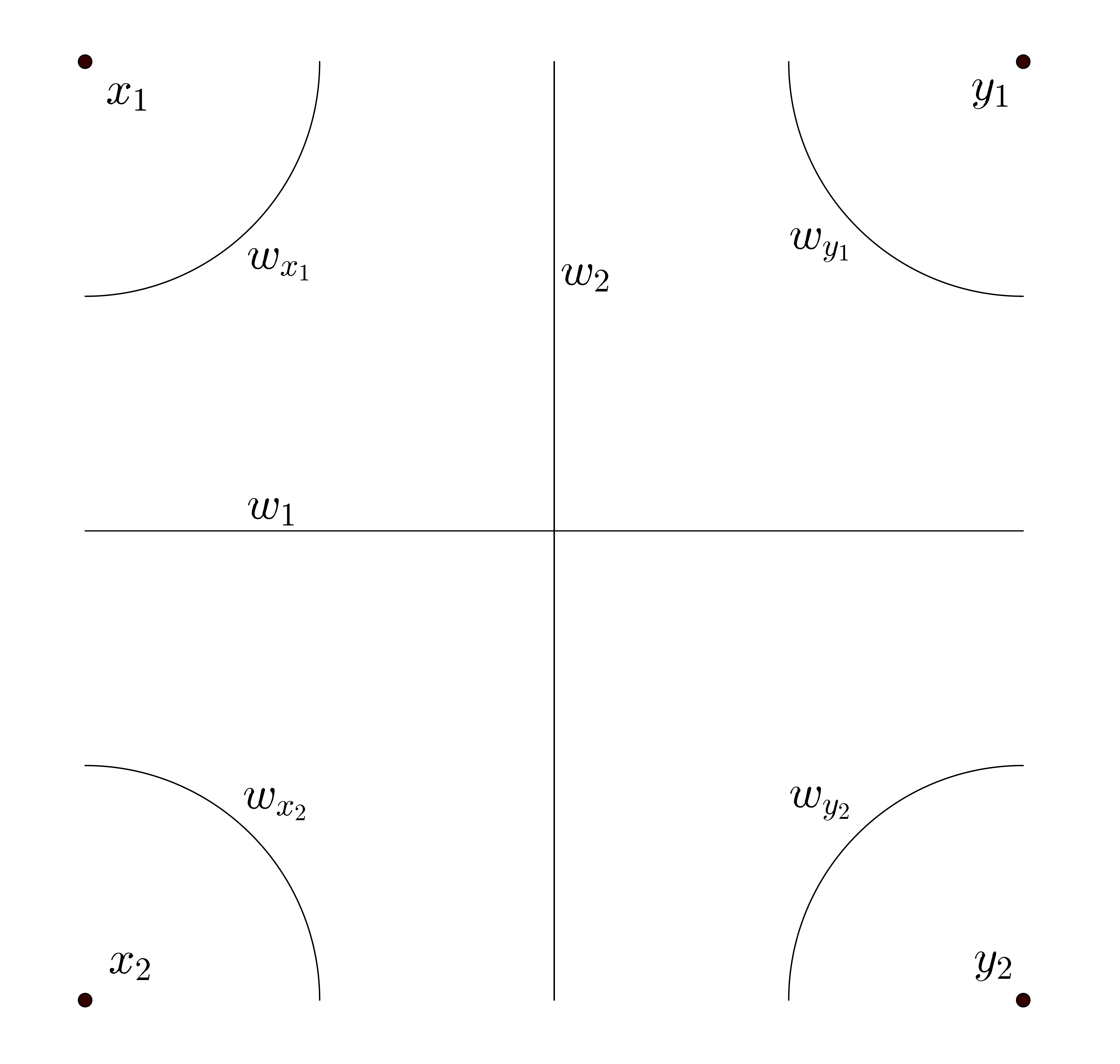}
	\caption{ Notation for Lemma~\ref{OnlyOneSideLemma}}
	\label{DistanceGraph}
\end{figure}
Then 
\begin{eqnarray*}
	d(x_1, y_1) = w_{x_1} + w_{y_1} + w_2 + w_3 &\quad& d(x_2, y_2) = w_{x_2} + w_{y_2} + w_2 + w_3 \\
	d(x_1, x_2) = w_{x_1} + w_{x_2} + w_1 + w_3 &\quad& d(y_1, y_2) = w_{y_1} + w_{y_2} + w_1 + w_3 \\
	d(x_1, y_2) = w_{x_1} + w_{y_2} + w_1 + w_2 &\quad& d(y_1, x_2) = w_{y_1} + w_{x_2} + w_1 + w_2
\end{eqnarray*}
\begin{proof}[Proof of Lemma~\ref{OnlyOneSideLemma}]
	We use the notation above.
	Let $x_1 = x$, $y_1 = y$, $x_2 = z$, $y_2 = z'$. Notice that any split that separates $\{x=x_1, y=y_1\}$ and $\{z= x_2, z'=y_2\}$ is incompatible with $S$, which implies that $w_1\leq \Omega(S) \leq \Omega$. Therefore
	\begin{eqnarray*}
		d(z,x)+d(z,y ) &=& d(x_1, x_2) + d(x_2, y_1)\\
		&=&(w_{x_1}+w_{x_2}+w_1+w_3) + (w_{x_2} + w_{y_1} + w_1 + w_2) \\ 
		&=& 2w_{x_2} + (w_{x_1}+w_{y_1}+w_2+w_3) + 2w_1\\
		&\leq& 2(w_{x_2}+w_{y_2}+w_2+w_3) + (w_{x_1}+w_{y_1}+w_2+w_3) + 2w_1\\
		&=& 2d(x_2,y_2)+d(x_1,y_1) + 2w_1\\
		&=& 2d(z,z')+d(x,y) + 2w_1\\ 
		&\leq& 2d(z,z')+d(x,y)+2\Omega.
	\end{eqnarray*}
\end{proof}

\subsection{Mini Reconstruction}

\begin{proof}[Proof of Proposition~\ref{prop:mini-correct}]
	Let $x, y \in \Taxa$ be any distinct pair of taxa with $\hat{d}(x,y) \leq \ellipseC$.
	To show that Mini Reconstruction correctly reconstructs the splits of $\Net$ restricted to $B(x,y)$, we first show that the diameter of $B(x,y)$ is at most $R$ and that therefore the distance matrix is accurate within it.
	\begin{lemma}[Diameter of $B(x,y)$]
		\label{BallIsSmall}
		For all $z,z' \in B(x,y)$, it holds that $d(z,z')\leq \M < R$.
	\end{lemma}
	\begin{proof}
		Because $z\in B(x,y)$, we have $\hat{d}(z,x)+\hat{d}(z,y) \leq \ellipseR$. So both $\hat{d}(z,x)$ and $\hat{d}(z,x)$ are smaller or equal than $\ellipseR < R$. Hence, because $\hat{d}$ is a $(\tau,R)$-distorted metric, we have that $d(z,x) \leq \hat{d}(z, x) + \tau$ and $d(z,y) \leq \hat{d}(z, y) + \tau$, and therefore
		\begin{eqnarray*}
			d(z,x) + d(z,y) \leq \hat{d}(z,x)+\hat{d}(z,y) + 2\tau \leq \M
		\end{eqnarray*}
		Similarly,  $d(z',x)+d(z',y) \leq \M$. Putting all this together, because $d$ satisfies the triangle inequality (see e.g.~\cite{Book}), we have finally
		\begin{eqnarray*}
			d(z, z') &\leq& \frac{1}{2}[d(z,x)+d(z,y)+d(z',x)+d(z',y)] \\
			&\leq& \frac{1}{2}\cdot 2 \cdot (\M) \\
			&=& \M.
		\end{eqnarray*}
	\end{proof}
	
	It remains to show that the split
	decomposition method applied to $\hat{d}$ restricted
	to $B(x,y)$ correctly reconstructs $\AllSmallSplitsO$.
	We proceed by showing, first, that $\AllSmallSplitsO$
	is the set of all $d$-splits restricted to $B(x,y)$
	and, second, that $\AllSmallSplits$
	(the set of all $\hat{d}$-split over $B(x,y)$ which have isolation index larger than $2\tau$) is also equal to the set of all $d$-splits restricted to $B(x,y)$.
	The first claim follows essentially from Lemma~\ref{lemma:dsplits-weakcomp}, which says that
	the splits of a circular network are its $d$-splits. The second claim follows from the correctness
	of the split decomposition method (Lemma~\ref{lemma:split-decomp}) and an error analysis.
	\begin{enumerate}
		\item Because $\Splits$ is circular, so is the restriction $\AllSmallSplitsO$ of $\Splits$ to $B(x,y)$
		as intersections are smaller on the restriction. 
		For $S = \{S_1,S_2\} \in \Splits$, we let $S|_{B(x,y)} = \{S_1\cap B(x,y), S_2 \cap B(x,y)\}$ be the restriction
		of $S$ to $B(x,y)$ (if both sides are non-empty). Also,
		for $S' \in \AllSmallSplitsO$, define
		the restricted weight of $S'$ as follows
		\begin{eqnarray*}
			w|_{B(x,y)}(S') 
			= \sum_{S \in \Splits\,:\,S|_{B(x,y)} = S'} w(S),
		\end{eqnarray*}
		where the sum accounts for the fact that many full splits
		on $\Taxa$ can produce the same restricted split on $B(x,y)$.
		Then, for any $z, z'\in B(x,y)$,
		\begin{eqnarray*}
			d(z, z') = \sum_{S' \in\AllSmallSplitsO}w|_{B(x,y)}(S') \, \delta_{S'} (z,z').
		\end{eqnarray*}
		By Lemma~\ref{lemma:dsplits-weakcomp} applied
		to $(B(x,y),\AllSmallSplitsO,w|_{B(x,y)})$,
		we know that $\AllSmallSplitsO$ is the set of all $d$-splits in $B(x,y)$ and, further, 
		$\alpha_d(S') = w|_{B(x,y)}(S')$ 
		for all $S' \in \AllSmallSplitsO$ (where the isolation
		index is implicitly restricted to $B(x,y)$).

		\item By Lemma~\ref{lemma:split-decomp}, the split
		decomposition method applied to $\hat{d}$ restricted
		to $B(x,y)$ returns all $\hat{d}$-splits.
		Because $\hat{d}$ is a $(\tau,R)$-distorted metric and by \Cref{BallIsSmall}, for every pair of taxa in $B(x,y)$ the difference between $\hat{d}$ and $d$ is smaller than $\tau$. Hence, for any $x_1, x_2, y_1, y_2 \in B(x,y)$, it follows that
		\begin{eqnarray*}
			|(d(x_1,y_1)+d(x_2,y_2))-(\hat{d}(x_1,y_1)+\hat{d}(x_2,y_2))| < 2\tau,
		\end{eqnarray*}
		which in turn implies that
		\begin{eqnarray*}
			&&|\max\{d(x_1,y_1)+d(x_2,y_2), d(x_1,x_2)+d(y_1,y_2), d(x_1,y_2)+d(y_1,x_2)\} \\
			&& \qquad- \max\{\hat{d}(x_1,y_1)+\hat{d}(x_2,y_2), \hat{d}(x_1,x_2)+\hat{d}(y_1,y_2), \hat{d}(x_1,y_2)+\hat{d}(y_1,x_2)\}| < 2\tau.
		\end{eqnarray*}
		So by the definition of $\tilde{\alpha}_d$, 
		\begin{eqnarray*}
			|\tilde{\alpha}_d(\{\{x_1,y_1\},\{x_2,y_2\}\})  - \tilde{\alpha}_{\hat{d}}(\{\{x_1,y_1\},\{x_2,y_2\}\}) | < 2\tau.
		\end{eqnarray*}
		As a result, for any split $S'$ over $B(x,y)$ (not necessarily in $\AllSmallSplitsO$), we have
		\begin{eqnarray*}
			|\alpha_d(S')  - \alpha_{\hat{d}}(S') | < 2\tau.
		\end{eqnarray*}
		Hence, for any split $S'$ over $B(x,y)$, it holds that $\alpha_{\hat{d}}(S') > 2\tau$ only if $\alpha_d(S') > 0$, that is, if $S'$ is a $d$-split on $B(x,y)$. This shows that $\AllSmallSplits \subseteq \AllSmallSplitsO$. On the other hand, for any $S'\in\AllSmallSplitsO$, we have $\alpha_d(S') = w|_{B(x,y)}(S') \geq \epsilon > 4\tau$ and, therefore, we have $\alpha_{\hat{d}}(S') > 2\tau$. This shows that $\AllSmallSplitsO \subseteq \AllSmallSplits$. Combining both statements concludes the proof.
		
	\end{enumerate}

\end{proof}

\subsection{Bipartition Extension}

\begin{proof}[Proof of Proposition~\ref{prop:ext-correct}]
	Let $x, y \in \Taxa$ be any distinct pair of taxa with $\hat{d}(x,y) \leq \ellipseC$.
	By Proposition~\ref{prop:mini-correct},
	$\AllSmallSplits = \AllSmallSplitsO$.
	In particular, $\SmallSplits = \SmallSplitsO$.
	Let $S' \in \SmallSplitsO$.
	By Lemma~\ref{lemma:hoppability} and
	the fact that $\Delta + 2 \Omega < R$,
	Bipartition Extension always terminates.
	
	The next lemma is the key step to proving that
	Bipartition Extension is correct, that is, that
	it puts every $z$ outside of $B(x,y)$ on the ``right side
	of $S'$.'' In words, it shows that jumping over a split requires a long hop. It follows from the key distance estimate in Lemma~\ref{OnlyOneSideLemma}.
	\begin{lemma}[Jumping over a split requires a long hop]
		\label{OtherSideIsFar}
		Suppose $x,y\in \Taxa$ satisfies $\hat{d}(x,y)\leq \ellipseC$ and $S=\{S_x,S_y\}\in\BigSplitsO$ is a split that separates $x$ and $y$. If $z\in S_x-B(x,y)$, then $\forall z' \in S_y$, $\hat{d}(z,z') > \connectD$.
	\end{lemma}
	\begin{proof}
		Since $z\in\Taxa-B(x,y)$,  we know that $\hat{d}(z,x)+\hat{d}(z,y) > 3\Delta + 7\Omega + 8\tau $. Hence, $d(z,x)+d(z,y) > 3\Delta + 7\Omega + 6\tau$. Indeed, if $d(z,x) + d(z,y) \leq 3\Delta + 7\Omega + 6\tau < R$, then both $d(z,x)$ and $d(z,y)$ would be less than $R$ and, by definition of a distorted metric,
		we would have
		\begin{eqnarray*}
			\hat{d}(z,x) + \hat{d}(z,y) \leq (d(z, x)+\tau) + (d(z,y)+\tau) \leq 3\Delta+7\Omega+ 8\tau,
		\end{eqnarray*}
		which would cause a contradiction.  Moreover, because $z\in S_x$ and $z'\in S_y$, we know that $S$ separates $\{z, x\}$ and $\{z', y\}$. Hence, by \Cref{OnlyOneSideLemma}, 
		$$
		d(z, z') \geq \frac{1}{2}(d(z,x)+d(z,y)-d(x,y)) - \Omega.
		$$ 
		In addition, since we chose $x,y$ such that $\hat{d}(x,y)\leq \ellipseC < R$, we know that $d(x,y) \leq \Delta+\Omega + 2\tau$. Combining all the above, we get	
		\begin{eqnarray*}
			d(z,z') &\geq& \frac{1}{2}(d(z,x)+d(z,y)-d(x,y)-2\Omega)\\
			&>& \frac{1}{2}[(3\Delta + 7\Omega + 6\tau ) - (\Delta+\Omega + 2\tau) - 2\Omega]\\
			&=& \frac{1}{2}(2\Delta+4\Omega+4\tau) = \Delta + 2\Omega + 2\tau
		\end{eqnarray*}
		
		which implies that $\hat{d}(z,z') > \Delta+2\Omega+2\tau - \tau = \connectD$.
	\end{proof}
	
	Note that every split in $S \in \Splits|_{x,y}$ has a
	non-trivial restriction $S'$ in $\SmallSplitsO$
	as $x \in S_x$ and $y \in S_y$.
	It remains to prove two things: that every $S' \in \SmallSplitsO$
	has a {\em unique} extension in $\Splits|_{x,y}$;
	and that Bipartition Extension correctly reconstructs it.
	We argue both points simultaneously by contradiction. Let
	$S \in \Splits|_{x,y}$ be an extension of $S'$, 
	that is, $S|_{B(x,y)} = S^*|_{B(x,y)} = S'$. 
	Suppose that the output of Bipartition Extension applied
	to $S'$ is {\em not} $S$ and let $z$ be the first taxon that is placed on the wrong side. That is, at the moment we are adding $z$, the split $\{\tilde{S}_x, \tilde{S}_y\}$ we are enlarging still satisfies $\tilde{S}_x\subseteq S_x$ and $\tilde{S}_y\subseteq S_y$ but, without loss of generality, $z\in S_x$ is added to $\tilde{S}_y$. The reason that we add $z$ to $\tilde{S}_y$ needs to be that there exists $z' \in \tilde{S}_y$ such that $\hat{d}(z,z')\leq\connectD$. Notice that, because $z\notin \tilde{S}_x\cup\tilde{S}_y$, $z$ is not in $B(x,y)$. Hence, $z\in S_x-B(x,y)$. But $z'\in S_y$ and satisfies $\hat{d}(z,z') \leq \connectD$---which contradicts the statement of~\Cref{OtherSideIsFar}. Hence, the output must be $S$ and that concludes the proof.
\end{proof}

\subsection{Wrapping up the proof}

\begin{proof}[Proof of Proposition~\ref{prop:exhaust}]
	That follows from Lemma~\ref{WitnessOfDepth} and
	the fact that $\Delta + \Omega < R$.
\end{proof}

\begin{proof}[Proof of Proposition~\ref{prop:weight-estimates}]
	Let	$S \in \BigSplits = \BigSplitsO$ and let $S' = S|_{B(x,y)}$ be the corresponding split on $B(x,y)$.
	As we have shown in the proof of Proposition~\ref{prop:mini-correct}, we have 
	$\alpha_d(S') = w|_{B(x,y)}(S')$ and $|\alpha_d(S')  - \alpha_{\hat{d}}(S') | < 2\tau$.
	Moreover, by the proof of Proposition~\ref{prop:ext-correct},
	$S$ is the unique extension of $S'$ to $\Taxa$
	and therefore $w|_{B(x,y)}(S') = w(S)$.
	The statement follows.
\end{proof}

\end{document}